\documentclass[a4paper,USenglish,cleveref,autoref]{lipics-v2021}

\usepackage{csquotes}

\usepackage{tabularx}
\usepackage{booktabs}

\usepackage{amsmath, amssymb, amsthm, mathtools, thmtools}
\usepackage{nicefrac}

\usepackage{xcolor}

\usepackage[ruled,vlined]{algorithm2e}

\usepackage{xspace}

\definecolor{mred}{RGB}{215,25,28}
\definecolor{mdarkblue}{RGB}{44,123,182}
\definecolor{mlightblue}{RGB}{171,217,233}
\definecolor{morange}{RGB}{253,174,97}

\usepackage{tikz}
\usetikzlibrary{decorations.pathreplacing}
\usetikzlibrary{math}
\usetikzlibrary{shapes,snakes}
\usetikzlibrary{patterns}

\DeclareMathOperator{\intSol}{Sol.int}

\DeclareMathOperator{\sens}{sens}
\DeclareMathOperator{\dist}{dist}
\DeclareMathOperator{\subDet}{subDet}

\newcommand{\jnew}{{\ensuremath{j_{\mathrm{new}}}}}

\newcommand{\Opt}{\operatorname{\text{\textsc{opt}}}}
\newcommand{\Alg}{\operatorname{\text{\textsc{alg}}}}
\newcommand{\pr}{{\textsc{CCS}}}
\newcommand{\jobs}{\mathcal{J}}
\newcommand{\machs}{\mathcal{M}}

\newcommand{\eps}{\varepsilon}

\newcommand{\phases}{\xi}

\DeclarePairedDelimiter\parenthesis{(}{)}

\DeclarePairedDelimiter\floor{\lfloor}{\rfloor}
\DeclarePairedDelimiter\ceil{\lceil}{\rceil}

\DeclarePairedDelimiter\set{\lbrace}{\rbrace}
\DeclarePairedDelimiterX\sett[2]{\lbrace}{\rbrace}{ #1 \,\delimsize| \,\mathopen{} #2 }

\title{Cardinality Constrained Scheduling in Online Models}

\titlerunning{Cardinality Constrained Scheduling}

\author{Leah Epstein}{Department of Mathematics, University of Haifa, Haifa, Israel}{lea@math.haifa.ac.il}{}{}

\author{Alexandra Lassota}{Chair of Discrete Optimization, EPFL, Lausanne, Switzerland}{alexandra.lassota@epfl.ch}{}{A. Lassota was supported by the Swiss National Science Foundation within the project \emph{Lattice algorithms and Integer Programming}~(200021\_185030/1).}
\author{Asaf Levin}{Faculty of Industrial Engineering and Management, The Technion, Haifa, Israel}{levinas@ie.technion.ac.il}{}{A. Levin was also partially supported by a grant from ISF - Israeli Science Foundation (grant number 308/18).}
\author{Marten Maack}{Heinz Nixdorf Institute \& Department of Computer Science, Paderborn University, Paderborn, Germany}{marten.maack@hni.uni-paderborn.de}{https://orcid.org/0000-0001-7918-6642}{Partially supported by the German Research Foundation (DFG) within the Collaborative Research Centre “On-The-Fly Computing“ under the project number 160364472 --- SFB 901/3.}
\author{Lars Rohwedder}{School of Business and Economics, Maastricht University, Maastricht, Netherlands}{l.rohwedder@maastrichtuniversity.nl}{https://orcid.org/0000-0002-9434-4589}{}

\authorrunning{Epstein et al.}


\category{}

\relatedversion{}

\supplement{}

\funding{This project was supported by GIF-Project "Polynomial Migration for Online Scheduling"}

\acknowledgements{}

\nolinenumbers 

\hideLIPIcs  

\EventEditors{John Q. Open and Joan R. Access}
\EventNoEds{2}
\EventLongTitle{42nd Conference on Very Important Topics (CVIT 2016)}
\EventShortTitle{CVIT 2016}
\EventAcronym{CVIT}
\EventYear{2016}
\EventDate{December 24--27, 2016}
\EventLocation{Little Whinging, United Kingdom}
\EventLogo{}
\SeriesVolume{42}
\ArticleNo{23}

\begin{document}

\maketitle

\begin{abstract}
Makespan minimization on parallel identical machines is a classical and intensively studied problem in
scheduling, and
a classic example for online algorithm analysis with Graham's famous list scheduling algorithm
dating back to the 1960s.
In this problem, jobs arrive over a list and upon an arrival, the algorithm needs to assign the job to a machine.
The goal is to minimize the makespan, that is, the maximum machine load.
In this paper, we consider the variant with an additional cardinality constraint:
The algorithm may assign at most $k$ jobs to each machine where $k$ is part of the input.
While the offline (strongly NP-hard) variant of cardinality constrained scheduling
is well understood and an EPTAS exists here,
no non-trivial results are known for the online variant.
We fill this gap by making a comprehensive study of various different online models.
First, we show that there is a constant competitive algorithm for the problem and further, present a lower bound of $2$ on the competitive ratio of any online algorithm.
Motivated by the lower bound, we consider a semi-online variant where upon arrival of a job of size $p$, we are allowed to migrate jobs of total size at most a constant times $p$. This constant is called the migration factor of the algorithm. Algorithms with small migration factors are a common approach to bridge the performance of online algorithms and offline algorithms. One can obtain algorithms with a constant migration factor by rounding the size of each incoming job and then applying an ordinal algorithm to the resulting rounded instance.
With this in mind, we also consider the framework of ordinal algorithms and characterize
the competitive ratio that can be achieved using the aforementioned approaches.
More specifically, we show that in both cases, one can get a competitive ratio that is strictly lower than
$2$, which is the bound from the standard online setting.
On the other hand, we prove that no PTAS is possible.
\end{abstract}

\newpage

\section{Introduction}
Scheduling jobs on identical parallel machines is a well-studied problem.
Such problems were in particular investigated extensively in online settings, where the algorithm has to make decisions before the whole instance is revealed.
Graham's List Scheduling from the 1960's~\cite{Gr66} is a textbook algorithm by now and an early example of an online algorithm (although the
notion of competitive analysis was not formalized at that time).
In this work, we study a generalization that considers an additional cardinality constraint on the number of
jobs allowed on a machine.

\subparagraph{The Cardinality Constrained Scheduling problem.}
We are given a set $\jobs$ of $n$ jobs, a set $\machs$ of $m$ identical parallel machines and a positive integer $k$.
Each job $j$ has a job size $p_j$, which is also known as the processing time of the job. A feasible solution is a non-preemptive schedule (each job has to be assigned as a whole) satisfying the condition that for each machine $i$, the number of jobs assigned to $i$ is at most $k$.  Our goal is to minimize the makespan, that is, the
maximum completion time of any job.
In the context of makespan minimization, one does not need to explicitly consider the time axis and instead,
a non-preemptive schedule can be defined as a partition of the job set to $m$ machines, that is, a function $\sigma: \jobs \rightarrow \machs$.
The {\em load} of machine $i$ in schedule $\sigma$ is the total size of jobs assigned to $i$, that is, $\sum_{j\in \sigma^{-1}(i)} p_j$.
The objective is to minimize the maximum load of a machine. It is easy to see that given $\sigma$, one can construct a schedule with makespan
equal to the maximum load. Summarizing, the goal for the cardinality constrained scheduling problem is to find a schedule $\sigma: \jobs \rightarrow \machs$ such that $\max_{i\in\machs} |\sigma^{-1}(i)| \leq k$ while minimizing the makespan $C_{\max}(\sigma) = \max_{i\in\machs}\sum_{j\in \sigma^{-1}(i)} p_j$.

The cardinality constraint arises naturally in settings where one needs to balance not only the loads of the machines, but also the number
of jobs. Suppose, for example, one wants to distribute passengers to airplanes for the same trip, but different times. The passengers' luggage weight
may vary and jet fuel usage is very sensitive to excess weight. Thus, the goal is to minimize the maximum loaded airplane (assuming
for simplicity that this dominates the fuel cost).
Extensions of the original problem to multiple dimensions are well-known and studied in both offline and online settings,
see for example the vector scheduling problem in~\cite{DBLP:journals/algorithmica/BansalOVZ16}.
However, in contrast to this problem, the second ``dimension'' in our problem is a hard constraint, which makes it much more difficult to handle.

In the offline setting, the job set is given beforehand and the goal is to find a feasible solution of minimum cost. We refer to \cite{LJ+16,DIMM06,DM01,He03,Ka21,KK13} for previous studies of the offline setting of the problem.
Since the problem is NP-hard in the strong sense, the best possible approximation result is an efficient polynomial time approximation scheme (EPTAS), that is, an algorithm that returns a feasible solution (if one exists) of cost at most $(1+\eps)$ times the optimal cost and the time complexity is upper bounded by the product of a computable function in $\eps$ and a polynomial in the (binary) encoding length of the input.
Such an algorithm was given in~\cite{LJ+16}.
Surprisingly, there exists no previous work on the online setting of this problem.

\subparagraph{Computational models studied in this work.}
In the online setting, the input is given as a sequence of jobs.  After a job is released, the algorithm learns the properties of the job (that is, the job size) and decides on the assignment of this job.  This assignment decision is irrevocable and the algorithm is forced to maintain the feasibility of the solution after the assignment of each job (as long as the input has a feasible solution).  
Once the job assignment is decided, the adversary constructing the input sequence learns the algorithm's assignment decision and chooses the size of the next job or stops the input sequence.
The {\em competitive ratio} of the online algorithm is a valid upper bound on the ratio between the cost of the solution returned by the algorithm and the optimal cost of an offline algorithm that sees the entire input sequence in advance (and may run in exponential time).

The model of ordinal algorithms is different.  Here, the algorithm needs to decide an assignment of $n$ jobs to $m$ machines without seeing the sizes of the jobs. The only information that the algorithm can access is
how these job sizes relate to each other, that is, the jobs are given as a list sorted non-increasingly by their sizes.
If the algorithm has decided upon the assignment $\sigma$, it means that the $i$-th largest job in the input sequence is assigned to machine $\sigma(i)$, and this applies for all $i$.
We say that an ordinal algorithm has {\em rate} $\alpha$ if for every input that satisfies the ordinal assumptions, the cost of the solution constructed by the assignment of the algorithm is at most $\alpha$ times the optimal cost for the same input.

Further, we study the model of algorithms with constant migration factor (also known as {\em robust algorithms}) similar to the online setting.
But, unlike in online algorithms, once a job~$j$ is released, it is also allowed to modify the schedule of a subset of jobs of total size at most $\beta \cdot p_j$ where $\beta$ is the migration factor.  We require that $\beta$ is a constant.  Usually, one cannot maintain an optimal solution using a robust algorithm. Thus, we use robust approximation algorithms.  We will use the terms competitive ratio and approximation ratio interchangeably, since some of our models are intermediate between online algorithms and offline algorithms. 
We say that a polynomial time algorithm that treats the input as a sequence is a robust $\alpha$-approximation algorithm if it has a constant migration factor and in every sequence of jobs, the resulting solution has cost of at most $\alpha$ times the optimal cost. 
Similarly, a robust PTAS is a family of algorithms containing a robust $(1+\eps)$-competitive algorithm for all $\eps>0$.
It is called robust EPTAS (or robust FPTAS respectively) if its running time is upper bounded by some computable function of (or a polynomial in) $\frac{1}{\eps}$ times a polynomial in the binary encoding length of the input.

\subparagraph{Results and outline of the paper.}

We present new results for all three of the models mentioned above.
An overview can be found in~\cref{cc:table:results}.
In the pure online case, we first prove a lower bound of $2$ on the competitive ratio of any (deterministic) algorithm.
A natural idea for an online algorithm is to create a balanced schedule, i.e., a schedule in which the property is maintained that any two machines receive approximately the same number of jobs.
This should limit the adversary's options to exploit the cardinality constraint.
However, we show that such an approach fails by establishing a lower bound of $m$ for the competitive ratio of algorithms maintaining the property that the number of jobs placed on any two machines differs by at most $o(\log(k))$. 
Another simple approach is to use variants of Greedy algorithms such as
the list scheduling algorithm, which always assigns the next job to the
machine with the lowest load. One would need to stop considering a machine
once it has received $k$ jobs.
However, this approach is also deemed to fail, since it may create a large
inbalance in the number of jobs assigned to the machines. If for example
one machine has only one job and all others are full (which could happen
using list scheduling), then the competitive ratio can be $k-1$
(when the next $k-1$ jobs are huge compared to the previous ones).

We utilize these insights in the design of an intricate online algorithm with constant competitive ratio, namely $120$.
This algorithm avoids both lower bounds by allowing a certain inbalance
in the number of jobs, which is then gradually reduced as more and more jobs
arrive. These results as well as a tight $\frac{1+\sqrt{5}}{2}$-competitive online algorithm for the special case $m=k=2$ are presented in \cref{sec:online}.

Next, we consider the mentioned relaxed online settings starting with ordinal algorithms in \cref{sec:ordinal}.
There is a known lower bound of $\frac{3}{2}$ regarding ordinal algorithm for the makespan minimization problem \cite{liu96} which applies to the \pr\ problem as well.
We present an ordinal algorithm with rate $\frac{81}{41}$ for \pr\ which is based on spreading out the $m$ largest jobs over all machines and then filling the machines gradually with a repeating overlapping pattern.
This gives an improvement over the rate $2$, which can be
achieved using a very simple round robin strategy.

In \cref{sec:migration}, we turn our attention to robust algorithms.
First, we show that an ordinal algorithm with rate at most $\alpha$ can be turned into a robust $((1+\epsilon)\alpha)$-approximation with migration factor $\frac{1+\epsilon}{\epsilon}$.
Together with our ordinal algorithm, this shows a separation between
the strict online setting
(having a lower bound of $2$) and the setting with migration.
On the other hand, we present a lower bound of roughly $1.05$ for the ratio of robust algorithms for \pr.
Hence, we cannot hope for a PTAS with a constant migration factor.
However, the lower bound only works for cases with $m\geq 3$ and $k$ part of the input, and we are able to present a robust EPTAS or FPTAS for the case with constant $k$ or $m=2, k > 1/\epsilon^2$ respectively.

We conclude this work in Section~\ref{sec:ClCS} by showing that the results of this paper cannot be extended to a generalization of \pr\ called Class Constrained Scheduling by showing a non-constant lower bound on the competitive ratio of robust algorithms for that problem. 
	\begin{table}
	\caption{An overview of the main results of this paper.}
	\centering
		\begin{tabular}{ll}
			\toprule
			Computational Model & Result \\ \midrule
			Online algorithms & $120$-competitive algorithm, lower bound of $2$ \\
			& Lower bound of $m$ for balanced algorithms\\
			& Tight $\frac{1+\sqrt{5}}{2}$-competitive algorithm for $m=k=2$\\\midrule			
			Ordinal algorithms & Algoritm with rate $\frac{81}{41}$\\\midrule
			Robust algorithms & Robust $((1+\epsilon)\cdot\frac{81}{41})$-approximation with $\frac{1+\epsilon}{\epsilon}$ migration factor \\
			& Lower bound of $\approx 1.05$ for constant migration, $m\geq 3$, unbounded $k$\\
			& Robust FPTAS for $m=2, k > 1/\epsilon^2$, and robust EPTAS for constant $k$\\
			\bottomrule
		\end{tabular}
		\label{cc:table:results}
	\end{table}

\subparagraph{Related work.}
The standard problem of makespan minimization on identical machines is obtained from the \pr\ problem by deleting the constraint saying that the number of jobs assigned to each machine is at most $k$.  At first glance, it seems that letting $k$ grow to infinity in \pr\ would lead to similar results to the ones known for makespan minimization on identical machines (without cardinality constraints). However, we show that this is not the case and the corresponding possible competitive ratios in our problem are significantly higher than the one achievable for the problem without cardinality constraints.  This is the case for the study of online algorithms as well as for robust algorithms.

The possible competitive ratio of the online algorithm for makespan minimization on identical machines is approximately $1.92$ \cite{Albers99,Fleischer2000} whereas the best lower bound for that problem is $1.88$ by Rudin \cite{RudinIII2001}.  
For small constant number of machines, it is known that there are better algorithms, for example, two machines List Scheduling (LS) \cite{Gr66} has a competitive ratio of $\frac{3}{2}$.  We establish a lower bound of $2-\frac 1k$ on the competitive ratio for \pr\ that shows that the possible competitive ratios for \pr\ are strictly higher than the ones achievable for the problem without cardinality constraints both in the regime of a small fixed number of machines and in the general case (in both of these scenarios, we establish a lower bound of $2$ when $k$ grows unboundedly).  Makespan minimization was also studied in terms of ordinal algorithms \cite{liu96} and it is known that there is an ordinal algorithm for this problem on identical machines with constant rate. 
In particular, for large numbers of machines $m$, there is an algorithm of rate at most $\frac 53$ and no algorithm has rate smaller than $\frac 32$.  The model of robust algorithms was introduced in \cite{SSS04} for makespan minimization on identical machines, where it is shown that there is a robust polynomial time approximation scheme for this problem.  Namely, for every $\eps>0$, there is a $(1+\eps)$-approximation algorithm whose migration factor is upper bounded by some function of $\eps$.

Cardinality Constrained Bin Packing (CCBP)  \cite{BCKK04,BBDEL_CCBP,Epstein05,KSS75} is the variant of \pr\ where the maximum job size is at most $1$, the makespan is forced to be at most $1$ in every feasible solution, but the algorithm is allowed to buy machines.  
The goal is to minimize the number of machines bought by the algorithm.
The best possible competitive ratio for CCBP is $2$ with respect to the absolute competitive ratio as well as with respect to the asymptotic competitive ratio  \cite{BCKK04,BBDEL_CCBP,BDE}.  
Regarding robust algorithms, it was shown in \cite{EL13} that for every fixed value of $k$ such that $k\geq 3$, there is no asymptotic approximation scheme for CCBP with constant migration factor.  Observe the difference with our results for \pr\ where for fixed constant value of $k$, we establish the existence of an approximation scheme for \pr\ with constant migration factor.

Ordinal algorithms were studied for other scheduling problems as well, see e.g. \cite{eps18survey,he02,liu96bp,tan01,tan05}, and robust algorithms were designed and analyzed for various scheduling problems and other packing problems (see e.g. \cite{BEJLMR19,BJK20,EL09,EL13,EL14,EL19,f+2018,GSV18,JK19,JKKL17,SV16}).

\subparagraph{Notation.} Throughout the paper, $\log$ refers to a logarithm with base $2$.  For a job subset $J$, we let $p(J)=\sum_{j\in J} p_j$.  For a positive integer $x$, we let $[x]=\{ 1,2,\ldots ,x\}$. Without loss of generality, we assume $\machs = [m]$.  When we consider a specific algorithm (online, ordinal, or an algorithm with constant migration), we let $\Alg$ denote the cost of the solution constructed by the algorithm, and we let $\Opt$ denote the optimal offline cost for the same instance.

\section{Pure online algorithms\label{sec:online}}
In this section, we study the competitive ratios of online algorithm for \pr.  We start by providing our lower bounds and then turn our attention to present a constant competitive algorithm for \pr.
Finally, we consider the special case $m=k=2$.

\subsection{Lower bound}

We show that when considering the online problem, there is no (deterministic) algorithm that has a competitive ratio smaller than $2$.

\begin{theorem}\label{pure2}
No online algorithm for \pr\ has a competitive ratio strictly smaller than $2$, and for a fixed value of $k$, no algorithm has a competitive ratio strictly smaller than $2-\frac 1k$.
\end{theorem}
\begin{proof}
We assume that $m\geq k$.
The input consists of two phases.  In the first phase, $m\times (k-1)$ jobs of size $1$ arrive.  Then the adversary examines the output of the algorithm.  If the algorithm has assigned exactly $k-1$ jobs to each machine, then the input continues with one big job of size $k$ that is the last job of the input.  In this case, we have that the cost of the algorithm is $\Alg = (k-1)+ k $, whereas the optimal offline cost is $\Opt = k$ as the offline solution could place at most $k$ unit sized jobs on $m-1$ machines and schedule the one big job on the remaining machine.  Observe that since $m\geq k$, we have $(m-1)\cdot k \geq m\cdot (k-1)$, so indeed there is such a schedule that assigns all unit sized jobs.

Next, consider the case that there is at least one machine that received at most $k-2$ jobs in phase 1.
Then the input continues with $m$ jobs of size $N$  where $N$ is some huge number. Then, at the end of these $m$ additional jobs, every machine is assigned exactly $k$ jobs, and by the condition of this case, we have $\Alg \geq 2N$. However, in this input, we can schedule on each machine exactly $k-1$ unit sized jobs together with one job of size $N$, and thus  $\Opt = N + k$.

In the first case, we have a ratio of $2 - 1/k$, and in the second case, the competitive ratio is at least $2N/(N+k)$.  Thus, when $k$ is a fixed constant, by letting $N$ grow unbounded, we get a lower bound of $2-\frac 1k$.  When $N=k$ and this common value grows unbounded, we get that the competitive ratio is at least $2$.
\end{proof}

We sometimes use the intuition of every machine having $k$ slots, each of which may contain exactly one job or may be empty.
Note that the ratio in the second case gets worse if a larger number of slots remain free on some machine, since the total number of free slots remaining after the arrival of the first phase is $m$.
This gives us the intuition that an algorithm should try to balance the number of jobs each machine receives.
However, this possible strategy cannot guarantee a constant competitive algorithm as the next proposition shows.
\begin{proposition}\label{prop2}
Let $t\geq 1$ be an integer number that may depend on $k$ such that $t=o(\log k)$.  Let $\Alg$ be an algorithm that maintains the invariant that the number of jobs placed on any two machines may differ by at most $t$.
Then the competitive ratio of $\Alg$ is at least $m$.
\end{proposition}
\begin{proof}
Assume by contradiction that there is a value of $m$ for which $\Alg$ has a competitive ratio $\rho(m) < m$. We select one such value of $m$ and see it as a constant.
Let $N$ be a large positive number such that $\frac{N-1}{N} = 1 -\frac 1N > \frac{\rho(m)}{m}$, and note that since $\rho(m)<m$, there exists such $N$ for which this inequality holds.
Let $k \geq 2$ be an integer such that $m \cdot  \frac{k \cdot \frac{N-1}{N}}{k+m\cdot N^{2mt -1}} > \rho(m)$, that is, $\frac{k \cdot \frac{N-1}{N}}{k+m\cdot N^{2mt -1}} > \frac{\rho(m)}{m}$.
Observe that when $k$ grows unbounded then $\frac{k \cdot \frac{N-1}{N}}{k+m\cdot N^{2mt-1}}$ tends to $\frac{N-1}{N}$ using the assumption that $t=o(\log k)$. By the condition on $N$, this is larger than the right hand side. Thus,  there exists an integer $k_0$ such that for all $k\geq k_0$ the required inequality holds, so we indeed can pick such a value of $k$.

Then, we consider the following input sequence.
The sequence ends once $k$ jobs are assigned to machine number $1$.
Further, the input sequence is split into $k$ rounds where
a round ends whenever a job is scheduled on machine number $1$. The job sizes within a round are chosen as follows.
The adversary presents jobs of size $1,N,N^2,N^3,\ldots$ (a geometric sequence of sizes that are different integer powers of $N$, starting with a job of size $1$), until the first time during this round when the current job is assigned to machine $1$ (the round ends when machine $1$ receives a job).  Observe that by the invariants of the algorithm, this event must happens after no more than $2 m\cdot t$ jobs of the round were released (this holds even if before the assignment the number of jobs of machine $1$ is much smaller than that of other machines and it becomes much larger).  Afterwards, a new round starts.  Observe that if the job assigned to $1$ is the $\ell$-th job of the round, then its size is $N^{\ell-1}$, whereas the total size of all jobs of this round is $\frac{N^\ell - 1}{N-1} < \frac{N^{\ell}}{N-1}$ that is smaller than $1+\frac{1}{N-1}$ times the size of the job assigned to $1$.  This applies for every round and therefore, the load of machine $1$ in the solution of $\Alg$ (to the entire instance) is larger than $\frac{N-1}{N}$ times the total size of jobs of the instance.

Denote by $X$ the total size of jobs of the instance, then we have $\Alg \geq X \cdot \frac{N-1}{N}$.
Let $p_{\max}$ be the maximum size of a job in the input sequence, then by the construction we conclude that $p_{\max} \leq N^{2mt-1}$, and we argue that there is an offline solution with cost at most $\frac{X}{m}+p_{\max}$.  In order to exhibit such an offline solution, consider the solution obtained by first sorting the jobs in a non-increasing order of their sizes, and then allocate the jobs of this sorted list in a round-robin fashion. Formally, the $i$-th job of the sorted list is allocated to machine of index $1+(i-1)\mod m$.  This round-robin allocation of jobs of non-increasing sizes has the property that if we consider the machine attaining the makespan, then by deleting the maximum sized job assigned to that machine, we get a load strictly smaller than the load of any other machine.  So the load of each machine in this round-robin solution is at most $\frac{X}{m}+p_{\max}$, and furthermore, the number of jobs assigned to each machine is exactly $k$.

Thus, the competitive ratio of $\Alg$ is at least $$\frac{X \cdot \frac{N-1}{N}}{\frac{X}{m}+p_{\max}} \geq \frac{X \cdot \frac{N-1}{N}}{\frac{X}{m}+N^{2mt -1}} \geq  \frac{k \cdot \frac{N-1}{N}}{\frac{k}{m}+N^{2mt -1}} = m \cdot  \frac{k \cdot \frac{N-1}{N}}{k+m\cdot N^{2mt -1}} > \rho(m) \ \ ,$$
  where the first inequality holds as $p_{\max} \leq N^{2mt -1}$, the second inequality holds as the function of $X$ defined as $\frac{X \cdot \frac{N-1}{N}}{\frac{X}{m}+N^{2mt-1}}$ is monotone increasing and $X\geq k$ using the fact that the first job of each round is a non-zero sized job and there are exactly $k$ rounds, and the last inequality holds by the condition on $k$. So, we get a contradiction to the assumption that the competitive ratio of $\Alg$ is smaller than $m$.
\end{proof}

Observe that obtaining an online algorithm with a competitive ratio of $\min\{ m,k\}$ is trivial, as any feasible solution has a cost that is at most $\min\{ m,k\}$ times the optimal cost. Thus, this is the competitive ratio of scheduling the jobs in a round-robin manner. Hence, the lower bound of Proposition~\ref{prop2} for this class of algorithms means that in order to establish small competitive algorithms, we need to exhibit an algorithm that does not belong to this class. In fact, in our algorithm, we have situations where there is a pair of machines with cardinalities that differ by $\Theta (\log k)$.

\subsection{A competitive algorithm for \pr}

Next, we present an algorithm for \pr\
with a constant competitive ratio.
For simplicity we assume that every job's size is a power of two, that is,
$p_j = 2^i$ for some (not necessarily positive) integer $i$.
More precisely, every incoming job is rounded down accordingly. Then the resulting makespan
(and the competitive ratio) will only increase by a factor of $2$ when considering the correct sizes.

\subparagraph{The general idea of the algorithm.}
We start by briefly discussing the main idea of the algorithm.
For simplicity of presentation, assume that we know the value $p^{\infty}_{\max}$, which is the maximum
size of any job by the end of the instance.

We group jobs by size. Group $G_i$ contains the jobs $j$ with $p_j = p^{\infty}_{\max} / 2^{i}$
for $i=0,\dotsc,\lfloor \log k \rfloor$. Further, group $G_{\infty}$ contains all smaller jobs,
that is, jobs of sizes below $p^{\infty}_{\max}/ 2^{\lfloor\log k \rfloor + 1} \le p^{\infty}_{\max} / k$
(recall that the jobs are rounded to powers of $2$).
Consider the following approach: Each group is scheduled independently using a round-robin strategy.
For each group, the first job of this size is assigned to machine $1$, the next one to machine $2$, etc.
Once every machine has one job of $G_i$, we continue with machine $1$ again for this group. This is done for all values of $i$ including $\infty$, see Figure~\ref{fig:log-competitive} for an illustration.
This method approximately balances both the loads of the machines and their cardinalities:
There are still differences between machine loads due to two reasons. 
First, each group may have one additional job (of the group) assigned to some machines compared to the other machines (this happens when the number of jobs of the group is not divisible by $m$). Second, the group $G_{\infty}$ may have jobs of very different sizes. All these sizes are small with respect to the maximum job size, and they are smaller by a factor of at least $k$ (so the total size of $k$ such jobs, which is the maximum per machine, is still at most $p^{\infty}_{\max}$).  Thus, the load of every machine is at most the average load plus an additional additive error term that is  at most
\begin{equation*}
  \sum_{i=0}^{\lfloor \log k  \rfloor} \frac{p^{\infty}_{\max}}{2^{i}} + k \cdot \frac{p^{\infty}_{\max}}{2^{\log k}}
  \le 3 p^{\infty}_{\max} .
\end{equation*}
We use $\Opt$ to denote the optimal makespan for the entire input. As $p^{\infty}_{\max}$ forms a lower bound on $\Opt$, and the average load is also a valid lower bound on $\Opt$, the makespan is never larger than four times the optimal makespan.
An issue arises once the cardinality on some machines arrives at $k$.
In that case it is no longer feasible to schedule all groups independently.
On the other hand, every pair of cardinalities (for two machines) differs by at most $\lfloor \log k  \rfloor + 2$.
Thus, if this difficulty occurs, then on each machine, there is only space for $O(\log k )$ more jobs (this is the number of remaining slots). If we assign the remaining jobs
arbitrarily, the difference in loads can increase only by $O(\log k ) \cdot p^{\infty}_{\max}$.
This implies we get a $O(\log k )$-competitive algorithm assuming we know the value $p^{\infty}_{\max}$.
\begin{figure}
  \centering
  \begin{tikzpicture}
    \draw[thick] (0.0-0.1, 3) -- (0-0.1, 0-0.1) -- (0.5+0.1, 0-0.1) -- (0.5+0.1, 3);
    \draw[thick] (1.0-0.1, 3) -- (1-0.1, 0-0.1) -- (1.5+0.1, 0-0.1) -- (1.5+0.1, 3);
    \draw[thick] (2.0-0.1, 3) -- (2-0.1, 0-0.1) -- (2.5+0.1, 0-0.1) -- (2.5+0.1, 3);
    \draw[thick] (3.0-0.1, 3) -- (3-0.1, 0-0.1) -- (3.5+0.1, 0-0.1) -- (3.5+0.1, 3);

    \draw[thick] (0, 0) rectangle (0.5, 0.5);
    \draw[thick] (1, 0) rectangle (1.5, 0.5);
    \draw[thick] (2, 0) rectangle (2.5, 0.5);
    \draw[thick] (3, 0) rectangle (3.5, 0.5);
    \draw[thick] (0, 0.5) rectangle (0.5, 1);
    \draw[thick] (1, 0.5) rectangle (1.5, 1);

    \draw[thick] (0, 0.0+1.1) rectangle (0.5, 0.25+1.1);
    \draw[thick] (1, 0.0+1.1) rectangle (1.5, 0.25+1.1);
    \draw[thick] (2, 0.0+1.1) rectangle (2.5, 0.25+1.1);
    \draw[thick] (3, 0.0+1.1) rectangle (3.5, 0.25+1.1);
    \draw[thick] (0, 0.25+1.1) rectangle (0.5, 0.5+1.1);
    \draw[thick] (1, 0.25+1.1) rectangle (1.5, 0.5+1.1);
    \draw[thick] (2, 0.25+1.1) rectangle (2.5, 0.5+1.1);
    \draw[thick] (3, 0.25+1.1) rectangle (3.5, 0.5+1.1);
    \draw[thick] (0, 0.5+1.1) rectangle (0.5, 0.75+1.1);

    \draw[thick] (0, 0.0+1.95) rectangle (0.5, 0.125+1.95);
    \draw[thick] (1, 0.0+1.95) rectangle (1.5, 0.125+1.95);
    \draw[thick] (2, 0.0+1.95) rectangle (2.5, 0.125+1.95);
    \draw[thick] (3, 0.0+1.95) rectangle (3.5, 0.125+1.95);
    \draw[thick] (0, 0.125+1.95) rectangle (0.5, 0.25+1.95);
    \draw[thick] (1, 0.125+1.95) rectangle (1.5, 0.25+1.95);
    \draw[thick] (2, 0.125+1.95) rectangle (2.5, 0.25+1.95);
    \draw[thick] (3, 0.125+1.95) rectangle (3.5, 0.25+1.95);
    \draw[thick] (0, 0.25+1.95) rectangle (0.5, 0.375+1.95);
    \draw[thick] (1, 0.25+1.95) rectangle (1.5, 0.375+1.95);
    \draw[thick] (2, 0.25+1.95) rectangle (2.5, 0.375+1.95);
    \draw[thick] (3, 0.25+1.95) rectangle (3.5, 0.375+1.95);
  \end{tikzpicture}
  \caption{Example schedule of the round-robin based algorithm}
  \label{fig:log-competitive}
\end{figure}
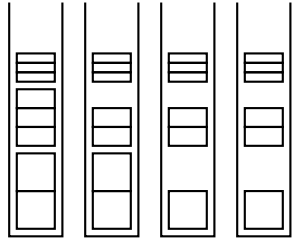

Indeed, the algorithm above is even $8$-competitive (including the factor of $2$ due to rounding)
as long as no machine's cardinality is full, that is, as long as no machine has $k$ jobs assigned to by the algorithm.
On the other hand, we showed in Proposition \ref{prop2} that no competitive algorithm can maintain a constant difference in the cardinalities of the job sets of any two machines, or even a difference of $o(\log k)$.
Nevertheless, we manage to obtain a constant competitive algorithm by modifying the approach stated above further.
Towards this we gradually decrease the number of groups $G_i$ as the machines get filled more and more,
so that by the time some machine is full (i.e., has been assigned $k$ jobs) there are only a constant number of groups and, in particular,
the cardinalities of the job sets assigned to any two machines differ only by a constant.
We will also remove the assumption of knowing $p^{\infty}_{\max}$ in advance.

\subparagraph{The algorithm.}
We first need to introduce some formal definitions.
We think of every machine as having $k$ slots, each of which may contain one job or it can be empty.
We form $k$ rows of these slots, where every row contains one slot of each machine.  A row is {\em full} if each slot of the row has a job, and it is {\em not full} otherwise, i.e., at least one slot of the row is empty.  A row is {\em free} if it has no job at all, i.e., all its slots are empty.
In the following,    let $p_{\max}$ denote  the maximum job size
of a job seen so far, and let $p^{\infty}_{\max}$ denote the maximum job size by the end of the instance.
We keep track of $p_{\max}$ in an online fashion in the sense that whenever a new job is released we check if we need to update (possibly increase) these values.

We form a partition of the jobs (released so far) into the {\em groups} $G_0,\dotsc,G_\ell, G_\infty$, where $\ell = \lfloor 2 \log(k) \rfloor$
and $G_i$, $i=0,\dotsc,\ell$, contains all jobs $j$ of size $p_j = p_{\max} / 2^i$.
Group $G_\infty$ contains all other (smaller) jobs.

The algorithm is defined recursively and in the recursive calls, it may remove full rows from the
existing schedule. The removal of the rows is only with respect to the rules of assigning
future jobs to the machines (in the sense that the jobs that were assigned to slots of these rows are still assigned to the corresponding machines and are not actually removed).
When rows are removed, and the number of rows is decreased, the value of $k$ will be decreased and the value of $\ell$ will be updated accordingly.
For a fixed value of $\ell$, we can maintain the partition into groups in an online fashion. The value of $\ell$ may become smaller, in which case some groups are merged into the group of small jobs. Since the jobs of each group except for $G_\infty$ share a common size, groups that are not merged remain unchanged.

To cope with the dynamic grouping where large jobs may later become small over time,
we change the previous algorithm in the following way.
Instead of keeping one row that is currently being filled for each group, we keep two.
One of these two rows may also contain jobs from $G_\infty$.
To make this more precise, we now formally state the structural invariants of the algorithm's schedule.

\subparagraph{The invariants of the algorithm.}
Consider the following structure in a schedule.
For each $i=0,\dotsc,\ell$ there are exactly two rows $r_i, r'_i$
containing elements of the group $G_i$. Row $r_i$ contains only elements
of $G_i$, whereas row $r'_i$ contains elements of $G_i$ and of $G_\infty$.
Moreover, of each pair of rows corresponding to a common group, at least one is not full.
Finally, there are $\lceil k / 2 \rceil - 2 (\ell + 1)$ rows containing
only elements from $G_\infty$ and none of them is full.
The remaining $\lfloor k/2 \rfloor$ rows are free.

The structure can be built trivially in the beginning when all rows are empty and
$2(\ell + 1) = 2\lfloor 2\log(k) \rfloor + 2 \le \lceil k/2 \rceil$, which holds for all $k \ge 49$.
In the case where $k$ is initially smaller, the algorithm will output an arbitrary feasible schedule,
which is $48$-competitive.

The definition of this structure may give the (false) impression that the algorithm tries to
keep the rows not full. This is not the case.
In its recursive calls, the algorithm removes certain rows from the instance when they become full
and the invariants above can also be read as: When two rows $r_i, r'_i$
(or one row belonging to $G_\infty$) become full, they need to be removed from the instance
and new rows (e.g., from the empty ones) need to be allocated to take their place.
When the algorithm removes a pair of rows (two rows corresponding to a common group become full),
it also decreases $k$ by $2$.

We will argue inductively that given such a structure, we can assign the remaining
jobs using the recursive algorithm in a way that maintains the invariants, and so that each
machine ends up with a total load (including the jobs already in the schedule) of at most
\begin{equation}\label{ldbound}
\frac{2}{m} \cdot p(J) + 8 \sum_{i=\log(p_{\max})}^{\log(p_{\max}^{\infty})} 2^i + \left(42 - \frac{1}{k-1} \right) p^\infty_{\max}.
\end{equation}
Notice that the second term is the sum of all job sizes (that is, all powers of $2$)
between $p_{\max}$ and $p^{\infty}_{\max}$ (multiplied with $8$).
Since the average load $p(J)/m$ and $1/2\cdot \sum_{i=\log(p_{\max})}^{\log(p_{\max}^{\infty})} 2^{i} \le p^{\infty}_{\max}$ form lower bounds on $\Opt$, this constitutes a $60$-competitive algorithm ($120$-competitive
when taking into account the initial rounding).
For $k \in \{48, 49\}$ it is trivial to see that there is an algorithm which (given the structure above)
assigns all remaining jobs online while maintaining the bound (\ref{ldbound}) on the loads.
This is because (\ref{ldbound}) is greater than
$49 p^{\infty}_{\max}$ and therefore any feasible schedule satisfies the bound. This proves the base case of our inductive argument.

\subparagraph{The algorithm for scheduling the next job while maintaining the invariants and satisfying the load bound (\ref{ldbound}).}
Let $h$ denote the number of free slots in our current schedule, that is, $k\cdot m$ minus the number of jobs
in the schedule.
Our induction is over $k$ and $h$:
If $k\in \{48, 49\}$, we observed that we can guarantee the makespan bound~\eqref{ldbound};
the base case $h = 0$ does not have to be considered, since this contradicts
the presumed structure on the current schedule (the induction will always end in $k \in \{48, 49\}$).
Assume now that $k \ge 50$ and that we have an algorithm that for all $k' \in \{k, k-1, k-2\}$ and $h' < h$
can continue the schedule with the specified structure while guaranteeing the bound~\eqref{ldbound}.
For all $k \ge 50$ it holds that $2 (\ell + 1) = 2 + 2 \lfloor 2 \log(k) \rfloor < \lceil k / 2 \rceil$.
This implies that the number of rows dedicated to $G_\infty$, 
$\lceil k / 2 \rceil - 2(\ell + 1)$, is strictly positive and the number of free rows,
$\lfloor k / 2 \rfloor$, is at least $25$, which will be important for our induction step.
Suppose some job $\jnew$ arrives.

First consider the case that $p_\jnew \le p_{\max}$ (referring to the value of $p_{\max}$ before
the arrival of $\jnew$)
and $\jnew \notin G_\infty$. Let $G_i$ be the group with $\jnew \in G_i$.
We assign $\jnew$ to an arbitrary empty slot
in $r_i$ or $r'_i$.
If one of the rows remains not full, we have maintained the structure and use
the induction hypothesis to prove that we can construct the remaining schedule
with the desired guarantee.
If on the other hand this makes both rows $r_i$ and $r'_i$ full, we first
remove the two full rows and then we are going to use the induction hypothesis as described below:
We set $k' := k - 2$.
This reduces the index of the last group to $\ell' = \lfloor 2\log(k')\rfloor$.
However, notice that $\ell \ge \ell' \ge \ell - 1$ (so it is possible that one group is merged into the group of small jobs).
\begin{description}
\item[Case 1: $\ell' = \ell$.] We remove one row dedicated to $G_\infty$ and pair it
with an empty row to form new rows $r'_i$ and $r_i$ respectively (we recall that the numbers of such rows are positive before the removal).
The number of rows dedicated to $G_\infty$ is now
$\lfloor k / 2 \rfloor - 2(\ell + 1) - 1 = \lfloor k'/2\rfloor - 2(\ell' + 1)$.
Hence, the structure is repaired and we can use the induction hypothesis.

\item[Case 2: $\ell' = \ell - 1$ and $i = \ell$.] Then group $G_i$ disappears,
the number of rows dedicated to $G_\infty$ is still
$\lfloor k/2 \rfloor - 2(\ell + 1) = \lfloor k' / 2 \rfloor - 2(\ell' + 1) - 1$.
To repair the structure, we add one empty row to the group $G_\infty$.

\item[Case 3: $\ell' = \ell - 1$ and $i < \ell$.] Then group $G_\ell$ is merged into $G_{\infty}$, and it creates two new rows for $G_\infty$ that were corresponding previously to $G_{\ell}$. We create new rows $r_i$ and $r'_i$ corresponding to $G_i$
by taking one of these rows previously corresponding to $G_{\ell}$ (the complete one if there is such a row) as $r'_i$
and an empty one as $r_i$.
This means the number of rows dedicated to $G_\infty$ is
$\lfloor k/2 \rfloor - 2(\ell + 1) + 1 = \lfloor k' / 2 \rfloor - 2 (\ell' + 1)$ and no row dedicated to $G_{\infty}$ is full (because full rows of $G_{\infty}$ will always be removed, and the only one that was perhaps temporarily added to the set of its rows was moved to the set of another group).
\end{description}
Let $J_C$ denote the jobs in the two full rows that we have just removed. By induction hypothesis,
we obtain a schedule for $J\setminus J_C$ where the load of each machine is at most
\begin{equation*}
\frac{2}{m} \cdot p(J\setminus J_C) + 8 \sum_{i=\log(p_{\max})}^{\log(p_{\max}^{\infty})} 2^{i} + \left(42 - \frac{1}{k' - 1}\right) p^\infty_{\max} .
\end{equation*}
Notice that the total size of the two jobs of $J_C$ on each machine is at least $p_{\max} / 2^i$ and
at most
$2 \cdot p_{\max} / 2^i \le 2/m \cdot p(J_C)$.
Thus, the total load on each machine is at most
\begin{multline*}
\frac{2}{m} p(J_C) + \frac{2}{m} p(J\setminus J_C) + 8 \sum_{i=\log(p_{\max})}^{\log(p_{\max}^{\infty})} 2^{i} + \left(42 - \frac{1}{k' - 1} \right) p^\infty_{\max} \\
\le \frac{2}{m} p(J) + 8 \sum_{i=\log(p_{\max})}^{\log(p_{\max}^{\infty})} 2^{i} + \left(42 - \frac{1}{k-1} \right) p^\infty_{\max} .
\end{multline*}
Now consider the case that $\jnew \in G_\infty$. We
add $\jnew$ to an arbitrary row dedicated to $G_\infty$. If the row remains
not full, we can directly use the induction hypothesis as the structure is maintained.
Otherwise, we remove the row and set $k' = k - 1$
(accordingly, $\ell' = \lfloor 2 \log(k') \rfloor$).
The number of rows dedicated to $G_\infty$ has reduced to
\begin{equation*}
  \lfloor k/2 \rfloor - 2 (\ell + 1) - 1
  \le \lfloor k'/2 \rfloor - 2(\ell + 1)
  \le \lfloor k'/2 \rfloor - 2(\ell' + 1) .
\end{equation*}
That means, we potentially have too few rows in $G_\infty$, but not too many.
On the other hand,
\begin{equation*}
  \lfloor k/2 \rfloor - 2 (\ell + 1) - 1
  \ge \lfloor k'/2 \rfloor - 2(\ell' + 2) - 1
  = \lfloor k'/2 \rfloor - 2(\ell' + 1) - 3 .
\end{equation*}
So $G_\infty$ is missing at most three rows.
We fill up these missing row with empty ones (recall there are at least $25$ empty rows)
and we have maintained the required structure.
Let again $J_C$ denote the jobs in the full row.  
Using the induction hypothesis and that
each job in $J_C$ is of size less than $p_{\max} / k^2$ (by using the previous values of $\ell$, this is the upper bound on the sizes of jobs of $G_{\infty}$),
we get a schedule with maximum load at most
\begin{align*}
&\frac{2}{m} p(J\setminus J_C) + 8 \sum_{i=\log(p_{\max})}^{\log(p_{\max}^{\infty})} 2^{i} + \left(42 - \frac{1}{k' - 1}\right) p^\infty_{\max} + \frac{1}{k^2} p_{\max} \\
&\le \frac{2}{m} p(J) + 8 \sum_{i=\log(p_{\max})}^{\log(p_{\max}^{\infty})} 2^{i} + \left(42 - \frac{1}{k-2} + \frac{1}{k^2} \right) p^\infty_{\max} \\
&\le \frac{2}{m} p(J) + 8 \sum_{i=\log(p_{\max})}^{\log(p_{\max}^{\infty})} 2^{i} + \left(42 - \frac{1}{k-1}\right) p^\infty_{\max} .
\end{align*}
Finally consider the case that $p_{\jnew} > p_{\max}$.
This means the new maximal job size is $p'_{\max} = p_{\jnew}$.
From the job set $J^0$ of the jobs in our current schedule (not including jobs of removed rows)
and excluding $\jnew$,
we construct a new job instance
${J^0}'$ where we increase the size of
the jobs in $G_i$, $i=0,\dotsc,\ell$, from $p_{\max}  / 2^i$ to $p'_{\max}  / 2^i$ and we consider
the same schedule for ${J^0}'$, which satisfies our required structure.
Let $J'$ be the union of jobs ${J^0}'$, $\jnew$, and all remaining jobs that have not arrived yet.
If we can assign all remaining jobs in this bigger instance $J'$ with
some bound on the loads of the machines, then we can in particular
obtain the same bound on the maximum load for our original instance $J$.
Notice that since in the current schedule there are only two rows containing jobs of each group $G_i$, $i<\infty$, there can be at most $2m$ jobs in total of each such group.
Thus,
\begin{equation*}
p(J') \le p(J) + 2m\sum_{i=0}^\ell \frac{p'_{\max}}{2^i}
\le p(J) +  8m \cdot \frac{p'_{\max}}{2} .
\end{equation*}
In the previous two cases, we have already proved that (for the given numbers $h$ and $k$) we
can schedule the remaining jobs of $J'$ with a bounded makespan,
since we are in the case that $p_{\jnew} \le p'_{\max}$. More precisely,
we obtain with the previous arguments a schedule for $J'$ (and in particular for $J$)
with a maximum load of at most
\begin{equation*}
\frac{2}{m} p(J') + 8 \sum_{i=\log(p'_{\max})}^{\log(p_{\max}^{\infty})} 2^{i} + \left(42 - \frac{1}{k-1}\right) p^\infty_{\max}
\le \frac{2}{m} p(J) + 8 \sum_{i=\log(p_{\max})}^{\log(p_{\max}^{\infty})} 2^{i} + \left(42 - \frac{1}{k-1} \right) p^\infty_{\max} ,
\end{equation*}
where the inequality holds since $\log(p'_{\max}) \geq \log(p_{\max})+1$ due to the rounding.
This concludes the proof.
\begin{theorem}
  For cardinality constrained scheduling there is a $120$-competitive online algorithm.
\end{theorem}

\subsection{The case $\boldsymbol{m=k=2}$}
We provide an online algorithm for the special case $m=k=2$. Note that for this case, there are at most four jobs, and we assume that there are exactly four jobs (by adding jobs of size zero). We use the definition $\varphi=\frac{1+\sqrt{5}}2$, where $\varphi-1=\frac 1{\varphi}$ and $\varphi^2=\varphi+1$. While this is a very small case, we can show tight bounds of $\varphi$ for it, while the bound of Theorem \ref{pure2} for this case is $1.5$.
An optimal solution will assign the largest job with the smallest job to one machine, and the two other jobs to the other machine.  

The input consists of jobs $1,2,3,4$ of sizes $p_1$, $p_2$, $p_3$, and $p_4$.
The algorithm is defined as follows. Assign jobs $1$ and $2$ to different machines. Assume without loss of generality that $p_1 \geq p_2$ (otherwise, the roles of jobs $1$ and $2$ are swapped). If $p_3 \leq \frac{p_1}{\varphi}$, assign job $3$ together with job $1$, and otherwise, assign it with job $2$. Assign job $4$ to the machine that has one job.

\begin{theorem}
The competitive ratio of the algorithm is at most $\varphi$, and this is the best possible competitive ratio.
\end{theorem}
\begin{proof}
The optimal solution has the following cost: $\min\{\max\{p_1+p_2,p_3+p_4\},\{p_1+p_3,p_2+p_4\},\{p_1+p_4,p_2+p_3\}\}$, which we denote by $\lambda$.
In particular, its cost is at least $\max\{p_1,p_2,p_3,p_4\}$. We will refer to the three options for the cost of the optimal solution in its cost as the three forms of the optimal solution.
For the analysis of the upper bound, we distinguish two cases.

\noindent{\bf Case 1: $\boldsymbol{p_3 \leq \frac{p_1}{\varphi}}.$}
In this case, the cost of the algorithm is $\max\{p_1+p_3,p_2+p_4\}$. We have $p_1+p_3 \leq p_1+\frac{p_1}{\varphi}= p_1\cdot (1+\frac{1}{\varphi})=\varphi\cdot p_1 \leq \varphi\cdot \lambda$, and it is left to find an upper bound for $p_2+p_4$.
If the optimal solution has the first form, we have $\lambda \geq p_1+p_2 \geq 2\cdot p_2$, and thus $p_2\leq \frac{\lambda}2$. We get $p_2+p_4 \leq \frac{\lambda}2+\lambda =1.5 \cdot \lambda < \varphi\cdot \lambda$.
If the optimal solution has the second form, we simply have $p_2+p_4 \leq \lambda$. If the optimal solution has the third form, we have $p_2+p_4 \leq p_1+p_4 \leq\lambda$, by $p_2\leq p_1$.

\noindent{\bf Case 2: $\boldsymbol{p_3 > \frac{p_1}{\varphi}}.$} In this case, the cost of the algorithm is $\max\{p_1+p_4,p_2+p_3\}$.
If the optimal solution has the third form, we find that the output is optimal.
If the optimal solution has the first form, we again use $p_2 \leq \frac{\lambda}2$ to get $p_2+p_3 \leq \frac{\lambda}2+\lambda < \varphi\cdot \lambda$. Additionally, we use $\lambda \geq p_3+p_4 > \frac{p_1}{\varphi}+p_4$.
Therefore, $p_1+p_4 \leq p_1+\varphi \cdot p_4 < \varphi\cdot \lambda$. 
If the optimal solution has the second form, we have $p_2+p_3 \leq p_1+p_3 \leq \lambda$. Additionally, by $p_3 > \frac{p_1}{\varphi}$ and $\lambda \geq p_1+p_3$, we have $\varphi \cdot p_1 =p_1+\frac{p_1}{\varphi} < p_1+p_3 \leq \lambda$, and thus $p_1+p_4 \leq \frac{\lambda}{\varphi}+\lambda=\varphi\cdot \lambda$.

\medskip

For the lower bound, let $M>0$ be a large number such that $2M^2 > \varphi \cdot (M+M^2)$. The input starts with two jobs of sizes $M$ and $1$. If the algorithms assigns them to the same machine, there are two further jobs of sizes equal to $M^2$. The algorithm has to assign them to the other machine for a makespan of $2\cdot M^2$. An optimal solution has makespan $M+M^2$, which gives us the required lower bound by the assumption on $M$.
If the first two jobs are placed on different machines, the next job has size $(\varphi-1)\cdot M$. We consider two cases for the assignment of the third job, and the fourth job will have size $1$ or $\varphi\cdot M$. The optimal cost for the first size of the fourth job is $M+1$, and for the second size it is $\varphi\cdot M+1$ (which is found by the simple structure of an optimal solution).

If the third job is assigned with the job of size $M$, the size of the fourth job is $1$, the cost of the algorithm is $M+ (\varphi-1)\cdot M =\varphi\cdot M$, and the competitive ratio is $\frac{\varphi}{1+ 1/M}$, which tends to $\varphi$ for $M$ growing to infinity. Otherwise, the size of the fourth job is $\varphi\cdot M$, the makespan of the algorithm is $M+\varphi\cdot M = (\varphi+1)\cdot M =\varphi^2\cdot M$. The competitive ratio is $\frac{\varphi^2}{\varphi+1/M}$, which also tend to $\varphi$ for $M$ growing to infinity.
\end{proof}

\section{Ordinal algorithms\label{sec:ordinal}}

Remember that in the ordinal setting, all jobs are given at the beginning in an ordered list, sorted non-increasingly by their sizes. 
However, the actual sizes are not included in the input.
A straight-forward approach in this setting is to assign the jobs via round-robin, i.e., the $i$-th largest job is assigned to machine $((i-1) \bmod m)+1$.
Note that this already gives an ordinal algorithm of rate $2$ and applies to both, the makespan minimization problem on identical machines as well as our problem. 
Thus, the goal of this section is to show that by delicately defining a different algorithm, we can get an ordinal algorithm of rate strictly smaller than $2$.

\subparagraph{Preliminaries and easy cases.}

We can always assume that the job sequence has $n=m\cdot k$ jobs.  
If it has more jobs, then we can safely output that there is no feasible solution, and otherwise, we can add $n-m\cdot k$ jobs of size $0$ at the end of the input sequence.
These zero sized jobs do not change the optimal cost, and for every feasible solution of the original instance, there is a corresponding feasible solution of the same cost with the zero sizes jobs (by adding for each machine the number of jobs so that it will have exactly $k$ jobs).

Furthermore, we will assume that $m\geq 2$ and $k\geq 3$.
For one machine ($m=1$), placing the $k$ input job on the single given machine is a trivial ordinal algorithm with rate $1$.
Similarly, if $k=1$, assigning the $i$-th largest of the $m$ input jobs to machine $i$ is again ordinal and optimal.
Lastly, in the case $k=2$, an optimal schedule can be achieved by assigning the first (largest) job to the first machine, the second to the second and so on until each machine received one job.
Afterwards, we change the direction, that is, job $m+1$ is assigned to machine $m$, job $m+2$ to machine $m-1$, and so forth.
This algorithm takes only the relative job sizes into account and is therefore ordinal.
Now consider the case that the maximum load is realized on a machine $i^*$ that receives two jobs with sizes $p\geq p'$.
Due to the assignment pattern, each predecessor machine of $i^*$ receives one job of size at least $p$ and each successor machine two jobs of size at least $p'$.
Hence, in an optimal solution there has to be a machine that receives both a job of size at least $p$ and a job of size at least $p'$ and our solution is therefore already optimal.

\subparagraph{First ideas.}

Observe that assigning the first $m$ jobs to different machines is necessary in any algorithm of rate strictly smaller than $2$, as the ordered input might consist of $m$ jobs of size $1$ and $m(k-1)$ jobs of size $0$.
On the other hand, the mentioned round-robin approach behaves badly if we have one big and many small jobs.
For instance, if we have $m=k$, one job of size $k$, $k(k-1)$ jobs of size $1$, and the remaining jobs of size $0$. 
Then the first machine receives load $2k-1$ in the round-robin approach while there is a trivial solution with objective value $k$. 

Keeping these examples in mind, a first idea for an ordinal algorithm might be to spread out the first $m$ jobs and afterward place fewer jobs on the machines that received the largest jobs.
More concretely we could, for instance, place the first $m$ jobs as described and then alternately take $m$ and $\ceil{m/2}$ jobs (from the input sequence, i.e., in non-increasing order) and place them on all and the last $\ceil{m/2}$ machines, respectively.
This is, in fact, the central idea for the known \cite{liu96} ordinal algorithm with rate $\frac{5}{3}$ for the case without cardinality constraints.
But in our case, we need another strategy after the last $\ceil{m/2}$ machines each received $k$ jobs.
The most obvious idea at this point, would be to apply round robin to the remaining jobs and first $\floor{m/2}$ machines.
However, it is relatively easy to see that we can adapt the bad example for round-robin to the resulting algorithm by simply doubling the number of machines and hence, a more sophisticated approach is needed.

Now, the first step of the algorithm presented here is again to place the $i$-th biggest job to machine $i$ for $i\in[m]$.
Then we use the above approach of placing jobs alternately on a smaller and a larger part of the last $\ceil{m/2}$ machines and apply it repeatedly so that the machines are gradually filled up starting from the last machines.
In the following, this approach is described in detail.

\subparagraph{The algorithm for the general case.}

We define the assignment of the jobs, where we assign the job sequence in a non-decreasing order of their sizes. 
Our assignment rules do not consider the sizes of jobs, just their position in the sorted list so indeed, we define an ordinal algorithm for the problem.

The assignment procedure works in {\em phases}, each of which is composed of {\em rounds}.  
Let $\phases = \floor{\log m}+2$ be the number of phases. 
A round is defined as an interval of consecutive machines $[m_{\ell},m_{r}-1]$, where the two indexes $m_{\ell},m_r\in [m+1]$ are called {\em border machines} (the value $m+1$ is to allow that an interval ends with the last machine).
In a round defined by the interval $[m_{\ell},m_{r}-1]$, we assign the next $m_r-m_{\ell}$ jobs to the machines of this interval where the $i$-th largest job in this set of jobs that we assign in the round is assigned to machine $m_{\ell}+i-1$.

Next, we define a subset of the machines that are border machines of some round of the algorithm. 
We set \[\mu(i) = \floor[\bigg]{\frac{m}{2^{\phases - i}}}+1\] for each $i\in[\phases]$. 
The border machines of the rounds are always from the set $\sett{\mu(b)}{b\in[\phases]}$.
Note that $\mu(1) = 1$, $\mu(2) = 2$ and $\mu(\phases)= m+1$ for any value of $m$. 
Observe that the difference between two consecutive border machines, i.e., $\mu(i+1) - \mu(i)$, grows approximately as a geometric sequence.
We briefly consider some examples:
\begin{itemize}
\item If $m=\set{2,3}$, we have $\phases = 3$ and the three borders are $1$, $2$, and $3$ or $4$ respectively.
\item If $m\in\set{4,5,6,7}$, we have $\phases = 4$, and the third border is $3$ for $m=4,5$ and $4$ for $m=6,7$.
\item If $m = 2^q$ for some integer $q\geq 1$, we have $q+2$ phases and  $\set{1}\cup\sett{1+2^{i}}{i\in\set{0,\dots, q}} $ is the set of border machines.
\end{itemize}

We conclude that in order to define the assignment, we need to define the intervals of the rounds of every phase.
Like in the last chapter, we use the intuition that each machine has $k$ slots to be filled by jobs. 
Our algorithm works as follows: 

\begin{enumerate}
\item In the first phase, there is only one round with borders $\mu(1) = 1$ and $\mu(\phases) = m+1$.

\item In the second phase, we repeat the following until all the slots between $\mu(\phases  -1)$ and $\mu(\phases)$ are filled:
One round with borders $\mu(\phases-2)$ and $\mu(\phases)$, followed by two rounds with borders $\mu(\phases  -1)$ and $\mu(\phases)$ (or less rounds if each machine of this last interval has exactly $k$ jobs).

\item In phase $s\in\set{3,\dots,\phases-1}$, there are alternating rounds with borders $\mu(\phases-s)$, $\mu(\phases-s +2)$ and $\mu(\phases  -s+1)$, $\mu(\phases-s+2)$ respectively. There are as many rounds as are needed to fill all the slots of the interval between $\mu(\phases  -s+1)$ and $\mu(\phases-s+2)$.

\item In the last phase, each round has borders $\mu(1) = 1$ and $\mu(2) = 2$, so in each such round, we assign one job to machine $1$ (and no other jobs to other machines). The number of rounds of this phase is so that the resulting number of jobs assigned to machine $1$ in all phases is exactly $k$.
\end{enumerate}

\subparagraph{Preliminaries for the analysis.}
 Let $i^*$ be the machine with lowest index that receives maximum load by the algorithm and $s^*$ be the phase in which $i^*$ receives its last job. Since $k>1$, we know that $s^* >1$. If $i^*$ receives its last job in the last phase, that is, if $s^* = \phases$, we know that $i^* = 1$.
Otherwise, $i^*$ has to be the middle border machine of phase $s^*$, that is, $i^* = \mu(\phases - s^* + 1)$.
To see this, note that due to the assignment rule, the load on the machines between $\mu(\phases - s^* + 1)$ and $\mu(\phases - s^* + 2)$ is non-increasing, and all other machines (i.e., the machines not in this interval) receive their last job in another phase (using $k>2$).

Let $r\in[k]$ and $p_r$ be the size of the $r$-th job $i^*$ receives.
We can assume that each job has a size from $\sett{p_r}{r\in[k]}\cup\set{0}$.
More precisely, we modify the input sequence in a way that maintains the ordinal settings, the size of jobs assigned to $i^*$ using our rules are the same, and the size of every other job is not increased (so the rate of the assignment on this modified instance is not smaller). In this modified instance, each job that appears in the list before the first job on $i^*$ has size $p_1$, each job that appears after the last job on $i^*$ has size $0$, and for each $r\in[k-1]$ each job that appears in the sorted list of jobs after the $r$-th but before the $(r+1)$-th job on $i^*$ has size $p_{r+1}$.
Note that the load on $i^*$ is not changed by these assumptions and the makespan in the optimal schedule is at most as big as before.
We set $P = \sum_{j= 2}^{k} p_j$ and therefore have:
\[\Alg = \sum_{r\in[k]} p_r = p_1 + P  \ . \]

When considering the borders of the phases, we will use the following simple observation:
\begin{remark}\label{rem:exp_and_mod}
For each non-negative integers $x,y$, we have $\floor[\big]{\frac{x}{2^y}} = 2\floor[\big]{\frac{x}{2^{y+1}}} + \floor[\big]{\frac{x \bmod 2^{y+1}}{2^y}}$.
\end{remark}
\begin{proof}
We have:
\begin{align*}
\floor[\Big]{\frac{x}{2^y}} &
= \floor[\bigg]{\frac{\floor[\big]{\frac{x}{2^{y+1}}}2^{y+1} + x \bmod 2^{y+1}}{2^y}} \\
& = \floor[\bigg]{\floor[\Big]{\frac{x}{2^{y+1}}}2 + \frac{x \bmod 2^{y+1}}{2^y}}
= 2\floor[\Big]{\frac{x}{2^{y+1}}} + \floor[\Big]{\frac{x \bmod 2^{y+1}}{2^y}} \ .
\end{align*}
\end{proof}

\subparagraph{Analyzing  the number of jobs assigned in each phase to some fixed machine.}

Next, we want to analyze how many jobs each machine receives in which phase.
It is obvious that each machine receives one job in the first phase.
Moreover, the machines from $\floor[\big]{\frac{m}{2}}+1$ to $m$ receive all their remaining jobs in the second phase.
All the remaining machines receive their remaining jobs (the ones not assigned in the first phase) in two succeeding phases.

For each $s\in\set{2,\dots,\phases-1}$ (and given $k$) let $\iota(s,k)$ be the number of jobs the first border machine of phase $s$ that is machine $\mu(\phases-s)$ (and all the following up to the next border) receives in phase $s$.
If $k$ is considered fix in a given context, we write $\iota(s)$ instead of $\iota(s,k)$.
\begin{lemma}\label{lem:iota_pattern1}
For each $s\in\set{2,\dots,\phases-1}$, we have $\iota(s) = \floor[\big]{\frac{k-1}{3}}$ if both $(k-1)\bmod 3 = 1$ and $s \bmod 2 = 1$ holds, and $\iota(s) = \ceil[\big]{\frac{k-1}{3}}$ otherwise.
\end{lemma}
\begin{proof}
Let $q= \floor[\big]{\frac{k-1}{3}}$ and $r= (k-1)\bmod 3$, that is, $k-1 = 3q + r$.  We prove the claim by induction over all $s\geq 2$.
By definition, we have $\iota(2) = \ceil[\big]{\frac{k-1}{3}}$ and $\iota(s') = \ceil[\big]{\frac{k-1 - \iota(s'-1)}{2}}$ for $s' \in \set{3,\dots,\phases-1}$.
Hence, the proof for the case $s=2$ is complete and we consider the case $s>2$ where we assume that the claim holds for $s-1$.
We distinguish the cases $r=0,1,2$.

Assume that $r=0$.
By induction we know $\iota(s-1) = q$, and hence $\iota(s) = \ceil[\big]{\frac{3q - q}{2}} = q = \frac{k-1}{3} = \ceil[\big]{\frac{k-1}{3}}$.

Assume that $r=1$.
By induction we know $\iota(s-1) = q$ if $s$ is even and $\iota(s-1) = q + 1$ if $s$ is odd.
This yields if $s$ is even that $\iota(s) = \ceil[\big]{\frac{3q + 1 - q}{2}} = q+1 = \ceil[\big]{\frac{k-1}{3}}$ and otherwise $\iota(s) = \ceil[\big]{\frac{3q + 1 - q - 1}{2}} = q = \floor[\big]{\frac{k-1}{3}}$.

In the last remaining case, we assume that $r=2$.
By induction we know $\iota(s-1) = q + 1$, and hence $\iota(s) = \ceil[\big]{\frac{3q + 2 - q - 1}{2}} = q + 1 = \ceil[\big]{\frac{k-1}{3}}$.
\end{proof}
Hence, we also know how many jobs the first border machine of $s$ receives in phase $s+1$, namely $(k-1) - \iota(s,k)$.
We are interested in the following, whether this number is even or odd.
\begin{corollary}\label{cor:even_odd}
Let $s\in\set{2,\dots,\phases-1}$, the machines $\mu(\xi - s)$ to $\mu(\phases  -s+1) - 1$ receive an even number of jobs in phase $s+1$ if and only if $(k-1)\bmod 3 = 0$ or both $(k-1)\bmod 3 = 1$ and $s \bmod 2 = 0$.
\end{corollary}
\begin{proof}
We use the notation of the proof of the last lemma and in each case we compute the value of  $(k-1) - \iota(s,k)$ as a linear function of $q$.
\begin{itemize}
\item If $r=0$, we have $(k-1) - \iota(s) = 3q - q = 2q$ and this is an even number.
\item If $r=1$ and $s \bmod 2 = 0$, we have $(k-1) - \iota(s) = (3q + 1) - (q + 1) = 2q$ and this is an even number.
\item If $r=1$ and $s \bmod 2 = 1$, we have $(k-1) - \iota(s) = (3q + 1) - q = 2q + 1$ and this is an odd number.
\item If $r=2$, we have $(k-1) - \iota(s) = (3q + 2) - (q+1) = 2q + 1$ and this is an odd number.
\end{itemize}
\end{proof}

In the following, we distinguish between the values of $s^*$.
The two extreme cases (second phase and last phase) are easier and we will establish better bounds for them.
The harder cases are the intermediate cases.

\subparagraph{Last Job in Second Phase.}
We consider the case $s^*= 2$.
Remember that $P = \sum_{j= 2}^{k} p_j$.
Our goal is a proof of the form
\[ \Alg = p_1 + P = (1 - \alpha)p_1 + (\alpha p_1 + P) \leq  (1 - \alpha)p_1 + \beta \cdot \mathtt{Load}/m \leq (1+ \beta - \alpha)\Opt ,\]
where $\mathtt{Load} =\sum_j p_j$ is the total size of jobs in the instance with suitable parameters $\alpha, \beta$ and using the obvious bound $\Opt \geq \max\set{p_1, \mathtt{Load}/m}$.
Observe that the last sequence of inequalities hold for all $\alpha$ in $[0,1]$ with the suitable $\beta$, so our goal would be to choose $\alpha,\beta$ for which the first inequality holds and the resulting rate of $1+\beta - \alpha$ is as small as possible.

In this case, the borders of the rounds are $\floor[\big]{\frac{m}{4}}+1$, $\floor[\big]{\frac{m}{2}}+1$ and $m+1$, and we have $i^* = \floor[\big]{\frac{m}{2}}+1$.
We get the following machine loads (see also \cref{fig:pattern1_s=2}):
\begin{figure}
\centering
\begin{tikzpicture}
\pgfmathsetmacro{\MachW}{0.5}
\pgfmathsetmacro{\MachH}{0.2}
\pgfmathsetmacro{\MachGap}{2.5*\MachW}
\pgfmathsetmacro{\JobH}{0.3}
\pgfmathsetmacro{\MachLabelShift}{-2.4*\MachH}

\foreach \a/\b/\c/\d in {	0/0/0/$1$, 
							1/1/0.6*\MachW/$\floor[\big]{\frac{m}{4}} + 1$,
							2/2/0.6*\MachW/$\floor[\big]{\frac{m}{2}} + 1$,
							3/2/0/,
							4/3/0/$m$
						}
{
	\draw[thick] 	(\a*\MachW + \b*\MachGap,-\MachH) -- 
					(\a*\MachW + \b*\MachGap,0) -- 
					(\a*\MachW + \b*\MachGap + \MachW,0) node[midway,yshift= \MachLabelShift cm, xshift = \c cm]{\d} -- 
					(\a*\MachW + \b*\MachGap + \MachW,-\MachH);
}
\foreach \a/\b in {1/0,2/1,4/2}{
	\draw[thick] (\a*\MachW + \b*\MachGap,0) -- (\a*\MachW + \b*\MachGap + \MachGap,0) node[midway, yshift = -0.7*\MachH cm ] {$\dots$};
}
\foreach \a/\b/\aa/\bb/\xs/\xxs/\h/\c/\l in 	
{
	0/0/2/2/0.5/0.5/1/black/$p_1$,
	3/2/4/3/0.5/0.5/1/black/$p_2$,
	1/1/2/2/0.5/0.5/2/black/$p_2$,
	3/2/4/3/0.5/0.5/2/black/$p_3$,
	2/2/2/2/0.5/0.5/3/black/$p_3$,
	3/2/4/3/0.5/0.5/3/black/$p_4$,
	2/2/2/2/0.5/0.5/4/black/$p_4$,
	3/2/4/3/0.5/0.5/4/black/$p_5$,
	1/1/2/2/0.5/0.5/6/black/$p_{\tilde{k}}$,
	3/2/4/3/0.85/0.85/6/mdarkblue/$p_{\tilde{k}+1}$,
	2/2/2/2/0.3/0.3/7/mdarkblue/$p_{\tilde{k} + 1}$,
	3/2/4/3/0.85/0.85/7/mred/$p_{\tilde{k}+2}$,
	2/2/2/2/0.3/0.3/8/mred/$p_{\tilde{k} + 2}$
}
{
	\draw[color = \c] (\a*\MachW + \xs*\MachW +\b*\MachGap,\h*\JobH) node[rectangle, fill = white, inner sep=0pt, text = \c] {\l} -- (\aa*\MachW + \bb*\MachGap + \xxs*\MachW,\h*\JobH) node[rectangle, fill = white, inner sep=0pt, text = \c] {\l};
}
\foreach \a/\b/\h in 
{
	3.5/2.5/5
}
{
	\node at (\a*\MachW + \b*\MachGap,\h*\JobH + 0.3*\JobH) {\vdots};
}

\end{tikzpicture}
\caption{The load pattern in the case $s^* = 2$. 
We set $\tilde{k} = k - ((k+1) \bmod 3)$. 
The last (red) and second to last job (blue) on the machines $\floor{\frac{m}{2}} + 1$ to $m$ are only placed if $\tilde{k} \leq k - 1$ or $\tilde{k}= k - 2$, respectively.} 
\label{fig:pattern1_s=2}
\end{figure}
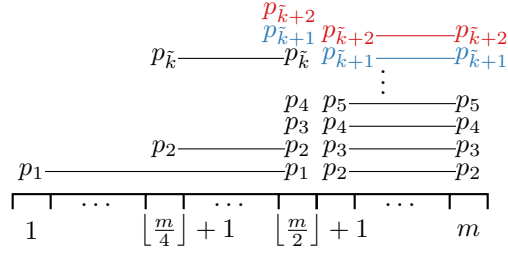
\begin{itemize}
\item Machines $1$ to $\floor[\big]{\frac{m}{2}}+1$ receive a job of size $p_1$.

\item Machines $\floor[\big]{\frac{m}{4}}+1$ to $\floor[\big]{\frac{m}{2}}$ additionally each receive the sequence $(p_2, p_5, \dots, p_{k - ((k+1)\bmod 3)})$. These jobs have an overall size of at least $P/3$. The other jobs assigned to these machines are zero sized.

\item Machines $\floor[\big]{\frac{m}{2}}+1$ to $m$ (additionally) receive the sequence $(p_2, p_3, \dots, p_{k})$ with overall size $P$ for each such machine.
\end{itemize}
Hence, we have:
\begin{align*}
\mathtt{Load} 	& \geq \bigg(\floor[\Big]{\frac{m}{2}} + 1\bigg)p_1 + \bigg(\floor[\Big]{\frac{m}{2}} - \floor[\Big]{\frac{m}{4}}  \bigg) \cdot \frac P3 + \bigg(m - \floor[\Big]{\frac{m}{2}}  \bigg) \cdot  P\\
				& = \bigg(\floor[\Big]{\frac{m}{2}} + 1\bigg)p_1 + \bigg(m - \frac{2}{3}\floor[\Big]{\frac{m}{2}} - \frac{1}{3}\floor[\Big]{\frac{m}{4}}  \bigg) P\\
				& =  \bigg(\floor[\Big]{\frac{m}{2}} + 1\bigg)p_1 + \bigg(m - \frac{5}{3}\floor[\Big]{\frac{m}{4}} - \frac{2}{3}\floor[\Big]{\frac{m \bmod 4}{2}}  \bigg) P
\end{align*}
In the last step, we used \cref{rem:exp_and_mod} for $x=m$, and $y=1$.
Furthermore, observe:
\begin{align*}
&\frac{1}{m}\parenthesis[\bigg]{m - \frac{5}{3}\floor[\Big]{\frac{m}{4}} - \frac{2}{3}\floor[\Big]{\frac{m \bmod 4}{2}} }
=  \frac{ \frac{7}{3}\floor[\big]{\frac{m}{4}} + (m \bmod 4) - \frac{2}{3}\floor[\big]{\frac{m \bmod 4}{2}} }{ 4\floor[\big]{\frac{m}{4}} + (m \bmod 4) }\\
= & \frac{7}{12} \frac{ \floor[\big]{\frac{m}{4}} + \frac{1}{4}\cdot (m \bmod 4) + \frac{5}{28}(m \bmod 4) - \frac{2}{7}\floor[\big]{\frac{m \bmod 4}{2}} }{ \floor[\big]{\frac{m}{4}} + \frac{1}{4}\cdot (m \bmod 4) }\geq \frac{7}{12}
\end{align*}
Hence, we have $$\frac{\mathtt{Load}}{m} \geq \frac{1}{2}\cdot p_1 + \frac{7}{12}\cdot P$$ and we can bound $\Alg$ as follows:
\[ \Alg = p_1 + P = \frac{1}{7}\cdot p_1 + \frac{12}{7}\cdot \Big(\frac{1}{2}p_1 + \frac{7}{12} P\Big) \leq \frac{13}{7}\cdot \Opt  \]

\subparagraph{Last Job in Intermediate Phase.}

We consider the case $s^*\in\set{3,\dots,\phases-1}$ and distinguish two subcases, namely, whether $s^*=3$ or not.
Note that the third phase may also be the last if $m\in\set{2,3}$ but here we assume that $s^*$ is not the last phase, so as long as we consider the intermediate phase, we assume that $m\geq 4$.

Let $\kappa = 2^{s^*-3}$, and thus $\kappa = 1$ if $s^* = 3$.
The borders of the rounds in phase $s^*$ are $\floor[\big]{\frac{m}{8\kappa}}+1$, $\floor[\big]{\frac{m}{4\kappa}}+1$, and $\floor[\big]{\frac{m}{2\kappa}}+1$.
Note that machine $i^* = \floor[\big]{\frac{m}{4\kappa}}+1$ receives jobs in three phases, namely phase $1$, $s^*-1$ and $s^*$.
Let $\iota^*$ be the number of jobs $i^*$ received at the end of phase $s^* - 1$, that is, $\iota^* = 1 + \iota(s^*-1,k)$ (see \cref{lem:iota_pattern1} and above).
Making use of \cref{lem:iota_pattern1}, we get the following load distribution (see also \cref{fig:pattern1_s=3} and \cref{fig:pattern1_s=4...xi-1})
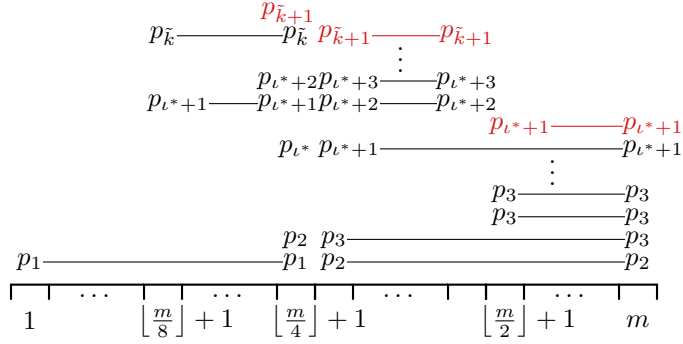
\begin{figure}[h]
\centering
\begin{tikzpicture}
\pgfmathsetmacro{\MachW}{0.5}
\pgfmathsetmacro{\MachH}{0.2}
\pgfmathsetmacro{\MachGap}{2.5*\MachW}
\pgfmathsetmacro{\JobH}{0.3}
\pgfmathsetmacro{\MachLabelShift}{-2.4*\MachH}

\foreach \a/\b/\c/\d in {	0/0/0/$1$, 
							1/1/0.6*\MachW/$\floor[\big]{\frac{m}{8}} + 1$,
							2/2/0.6*\MachW/$\floor[\big]{\frac{m}{4}} + 1$,
							3/2/0/,
							4/3/0.6*\MachW/,
							5/3/0.6*\MachW/$\floor[\big]{\frac{m}{2}} + 1$,
							6/4/0/$m$
						}
{
	\draw[thick] 	(\a*\MachW + \b*\MachGap,-\MachH) -- 
					(\a*\MachW + \b*\MachGap,0) -- 
					(\a*\MachW + \b*\MachGap + \MachW,0) node[midway,yshift= \MachLabelShift cm, xshift = \c cm]{\d} -- 
					(\a*\MachW + \b*\MachGap + \MachW,-\MachH);
}
\foreach \a/\b in {1/0,2/1,4/2,6/3}{
	\draw[thick] (\a*\MachW + \b*\MachGap,0) -- (\a*\MachW + \b*\MachGap + \MachGap,0) node[midway, yshift = -0.7*\MachH cm ] {$\dots$};
}
\foreach \a/\b/\aa/\bb/\xs/\xxs/\h/\c/\l in 	
{
	0/0/2/2/0.5/0.5/1/black/$p_1$,
	3/2/6/4/0.5/0.5/1/black/$p_2$,
	2/2/2/2/0.5/0.5/2/black/$p_2$,
	3/2/6/4/0.5/0.5/2/black/$p_3$,
	5/3/6/4/0.5/0.5/3/black/$p_3$,
	5/3/6/4/0.5/0.5/4/black/$p_3$,
	2/2/2/2/0.5/0.5/6/black/$p_{\iota^*}$,
	3/2/6/4/0.9/0.9/6/black/$p_{\iota^* + 1}$,
	5/3/6/4/0.9/0.9/7/mred/$p_{\iota^* + 1}$,
	1/1/2/2/0.9/0.3/8/black/$p_{\iota^* + 1}$,
	3/2/4/3/0.9/0.5/8/black/$p_{\iota^* + 2}$,
	2/2/2/2/0.3/0.3/9/black/$p_{\iota^* + 2}$,
	3/2/4/3/0.9/0.5/9/black/$p_{\iota^* + 3}$,
	1/1/2/2/0.5/0.5/11/black/$p_{\tilde{k}}$,
	3/2/4/3/0.8/0.5/11/mred/$p_{\tilde{k}+1}$,
	2/2/2/2/0.3/0.3/12/mred/$p_{\tilde{k}+1}$
}
{
	\draw[color = \c] (\a*\MachW + \xs*\MachW +\b*\MachGap,\h*\JobH) node[rectangle, fill = white, inner sep=0pt, text = \c] {\l} -- (\aa*\MachW + \bb*\MachGap + \xxs*\MachW,\h*\JobH) node[rectangle, fill = white, inner sep=0pt, text = \c] {\l};
}
\foreach \a/\b/\h in 
{
	5.5/3.5/5,
	3.5/2.7/10
}
{
	\node at (\a*\MachW + \b*\MachGap,\h*\JobH + 0.3*\JobH) {\vdots};
}
\end{tikzpicture}
\caption{The load pattern in the case $s^*=3$. 
There is some ambiguity marked in red.
The sixth row on machines $\floor[\big]{\frac{m}{2}}+1$ to $m$ represents $(k+1) \bmod 3$ many rows.
Moreover, we have $\tilde{k} = k-1$ if $(k-1)\bmod 3 =0$ or both $(k-1)\bmod 3 =1$ and $s^* \bmod 2 = 1$ hold.
Otherwise, $\tilde{k} = k$ and the last row on machines $\floor[\big]{\frac{m}{4}}+1$ to $\floor[\big]{\frac{m}{2}}$ is not placed.
} 
\label{fig:pattern1_s=3}
\end{figure}
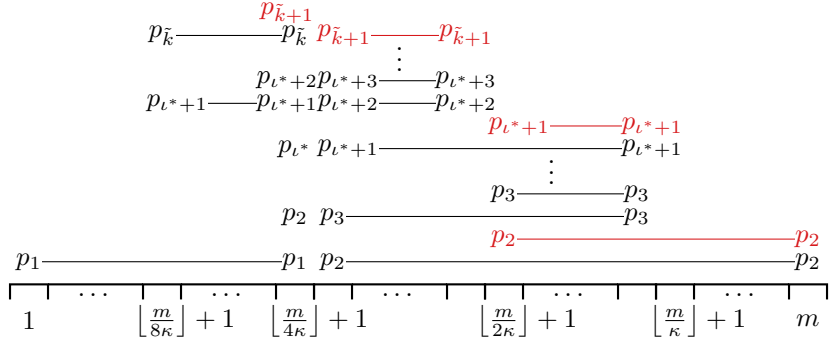
\begin{figure}[h]
\centering
\begin{tikzpicture}
\pgfmathsetmacro{\MachW}{0.5}
\pgfmathsetmacro{\MachH}{0.2}
\pgfmathsetmacro{\MachGap}{2.5*\MachW}
\pgfmathsetmacro{\JobH}{0.3}
\pgfmathsetmacro{\MachLabelShift}{-2.4*\MachH}

\foreach \a/\b/\c/\d in {	0/0/0/$1$, 
							1/1/0.6*\MachW/$\floor[\big]{\frac{m}{8\kappa}} + 1$,
							2/2/0.6*\MachW/$\floor[\big]{\frac{m}{4\kappa}} + 1$,
							3/2/0/,
							4/3/0.6*\MachW/,
							5/3/0.6*\MachW/$\floor[\big]{\frac{m}{2\kappa}} + 1$,
							6/4/0.6*\MachW/,
							7/4/0.6*\MachW/$\floor[\big]{\frac{m}{\kappa}} + 1$,
							8/5/0/$m$
						}
{
	\draw[thick] 	(\a*\MachW + \b*\MachGap,-\MachH) -- 
					(\a*\MachW + \b*\MachGap,0) -- 
					(\a*\MachW + \b*\MachGap + \MachW,0) node[midway,yshift= \MachLabelShift cm, xshift = \c cm]{\d} -- 
					(\a*\MachW + \b*\MachGap + \MachW,-\MachH);
}
\foreach \a/\b in {1/0,2/1,4/2,6/3,8/4}{
	\draw[thick] (\a*\MachW + \b*\MachGap,0) -- (\a*\MachW + \b*\MachGap + \MachGap,0) node[midway, yshift = -0.7*\MachH cm ] {$\dots$};
}
\foreach \a/\b/\aa/\bb/\xs/\xxs/\h/\c/\l in 	
{
	0/0/2/2/0.5/0.5/1/black/$p_1$,
	3/2/8/5/0.5/0.5/1/black/$p_2$,
	5/3/8/5/0.5/0.5/2/mred/$p_2$,
	2/2/2/2/0.5/0.5/3/black/$p_2$,
	3/2/6/4/0.5/0.5/3/black/$p_3$,
	5/3/6/4/0.5/0.5/4/black/$p_3$,
	2/2/2/2/0.5/0.5/6/black/$p_{\iota^*}$,
	3/2/6/4/0.9/0.9/6/black/$p_{\iota^* + 1}$,
	5/3/6/4/0.9/0.9/7/mred/$p_{\iota^* + 1}$,
	1/1/2/2/0.9/0.3/8/black/$p_{\iota^* + 1}$,
	3/2/4/3/0.9/0.5/8/black/$p_{\iota^* + 2}$,
	2/2/2/2/0.3/0.3/9/black/$p_{\iota^* + 2}$,
	3/2/4/3/0.9/0.5/9/black/$p_{\iota^* + 3}$,
	1/1/2/2/0.5/0.5/11/black/$p_{\tilde{k}}$,
	3/2/4/3/0.8/0.5/11/mred/$p_{\tilde{k}+1}$,
	2/2/2/2/0.3/0.3/12/mred/$p_{\tilde{k}+1}$
}
{
	\draw[color = \c] (\a*\MachW + \xs*\MachW +\b*\MachGap,\h*\JobH) node[rectangle, fill = white, inner sep=0pt, text = \c] {\l} -- (\aa*\MachW + \bb*\MachGap + \xxs*\MachW,\h*\JobH) node[rectangle, fill = white, inner sep=0pt, text = \c] {\l};
}
\foreach \a/\b/\h in 
{
	5.5/3.5/5,
	3.5/2.7/10
}
{
	\node at (\a*\MachW + \b*\MachGap,\h*\JobH + 0.3*\JobH) {\vdots};
}
\end{tikzpicture}
\caption{The load pattern in the case $s^*\in\set{4,\dots,\phases-1}$. 
There is some ambiguity marked in red.
Namely, the second row jobs on machines $\floor[\big]{\frac{m}{2\kappa}}+1$ to $m$ represent multiple placements of jobs of size $p_2$.
The sixth row on these machines may be present or not depending on $s^*$ and $k$.
Moreover, we have $\tilde{k} = k-1$ if $(k-1)\bmod 3 =0$ or both $(k-1)\bmod 3 =1$ and $s^* \bmod 2 = 1$ hold.
Otherwise, $\tilde{k} = k$ and $p_{\tilde{k}+1} = 0$.
} 
\label{fig:pattern1_s=4...xi-1}
\end{figure}
\begin{itemize}
\item Machines $1$ to $\floor[\big]{\frac{m}{4\kappa}}+1$ receive a job of size $p_1$.  Furthermore, all other jobs on machines at most $\floor[\big]{\frac{m}{8\kappa}}$ have zero size.

\item Machines $\floor[\big]{\frac{m}{8\kappa}}+1$ to $\floor[\big]{\frac{m}{4\kappa}}$ (additionally) receive the sequence $(p_{\iota^* +1 }, p_{\iota^* +3 }, \dots, p_{\tilde{k}})$ with $\tilde{k} = k-1$ if $(k-1)\bmod 3 =0$ or both $(k-1)\bmod 3 =1$ and $s^* \bmod 2 = 1$ hold (see \cref{cor:even_odd}), and $\tilde{k} = k$ otherwise.  The other jobs on these machines are of size $0$.

\item Each machine between $\floor[\big]{\frac{m}{4\kappa}}+1$ and $\floor[\big]{\frac{m}{2\kappa}}$ (additionally) receive jobs of total size $p_2+ p_3 + \dots + p_{k}$.  As for machines $\floor[\big]{\frac{m}{2\kappa}}+1$ to $m$, we distinguish between the two cases $s^*=3$ and $s^*>3$.

\item In the case $s^*=3$, $\floor[\big]{\frac{m}{\kappa}} = m$ so $i^*=\floor[\big]{\frac{m}{4}}+1$.  Therefore, until $i^*$ receives the $\iota^*+1$ job, these machines get three jobs of each size $i^*$ gets (where in the last size the number of rounds depend on the value of $k\bmod 3$).  Thus, machines $\floor[\big]{\frac{m}{2}}+1$ to $m$ receive the sequence $(p_2, p_3, \dots, p_{\iota^*}, p_{\iota^* + 1})$ with multiplicities $(1,3,\dots,3, 1+ ((k+1) \bmod 3))$.
Note that in the special case $k\in \set{3,4}$, we have $\iota^* = 2$ and the machines receive one job of size $p_2$ and $k-2$ jobs of size $p_3$.

\item In the case $s^*>3$, the machines $\floor[\big]{\frac{m}{2\kappa}}+1$ to $\floor[\big]{\frac{m}{\kappa}}$ receive the sequence $(p_2, p_3, \dots, p_{\iota^*}, p_{\iota^* + 1})$  with multiplicities $(1 + \iota(s^*-2,k),2,\dots,2, x)$ where $x=2$ if $(k-1)\bmod 3 = 0$ or both $(k-1)\bmod 3 = 1$ and $s^*\bmod 2 = 0$, and $x=1$ otherwise (see \cref{cor:even_odd}).  The difference from the case $s^*=3$ is because for $s^*>3$ each round in which $i^*$ gets a job is followed by only one round in which $i^*$ does not get a job (but these machines get additional job in this round) and this holds until these machines are filled with $k$ jobs each. Machines $\floor[\big]{\frac{m}{\kappa}}+1$ to $m$ receive $k$ jobs each of which of size $p_2$ as all of these were assigned to those machines before $i^*$ got its second job.
\end{itemize}

To make use of the assignment rules, we need an additional lower bound on $\Opt$.
To do so, we show that some of the smallest jobs placed on $i^*$ are in some sense also present in optimal solutions.
For this, we set $\psi(s,k) = \floor[\big]{\frac{\iota(s,k) - 2}{3}}$ for each $s\in\set{2,\dots,\phases-1}$ and write $\psi(s)$ instead of $\psi(s,k)$ if $k$ is fixed in the given context.
\begin{lemma}\label{lem:migration_pattern1_UB}
We have $\Opt\geq p_1 + \sum_{j=0}^{\psi(s^*)-1}p_{k-j}$.
\end{lemma}
\begin{proof}
Let $m'\leq m$, $\hat{p}$ be the total size of the $m'$ biggest jobs, and $\check{p}$ be the total size of the $(k-1)m'$ smallest jobs in the instance.
We want to utilize the (obvious) lower bound $\Opt\geq \frac{\hat{p} + \check{p}}{m'}$.  This is indeed a valid lower bound on $\Opt$ as there are $m'$ machines that are assigned  the $m'$ biggest jobs (in the optimal solution) and these machines process $(k-1)m'$ additional jobs of total size not smaller than $\check{p}$.
In particular, we will choose $m'$ such that $m' \in\set{ \floor[\big]{\frac{m}{4\kappa}},\floor[\big]{\frac{m}{4\kappa}}+1}$.
This yields $\hat{p} = p_1\cdot m'$, and we will show in the following $\check{p} \geq m'\sum_{j=0}^{\psi(s^*)-1}p_{k-j}$, thereby concluding the proof.

We consider the $(k-1)m'$ smallest jobs.
For this, we partition the jobs in the instance into $k+1$ groups.
The first group includes the first job that was placed on $i^*$ and all the jobs placed before (these have size $p_1$), for $\ell\in\set{2,\dots,k}$ group $l$ includes the $\ell$-th job placed on machine $i^*$ and all the jobs placed before that and do not belong to a preceding group (these have size $p_{\ell}$), and finally group $k+1$ contains the jobs placed after $i^*$ received its last job (these have size $0$).
Note that we can compute exactly how many jobs belong to each group as we do next.

We first determine the number of jobs $d$ of group $k+1$.
These are the jobs placed in the last round of phase $s^*$ on the machines $\floor[\big]{\frac{m}{4\kappa}}+2$ to $\floor[\big]{\frac{m}{2\kappa}}$ on the one hand, and the jobs placed in succeeding phases on the preceding machines. Therefore,
\[d = \parenthesis[\bigg]{\floor[\Big]{\frac{m}{2\kappa}} - \floor[\Big]{\frac{m}{4\kappa}} -1} + (k-1)\floor[\Big]{\frac{m}{4\kappa}} - \parenthesis[\bigg]{\floor[\Big]{\frac{m}{4\kappa}} - \floor[\Big]{\frac{m}{8\kappa}}}\iota(s^*)\]

Next, we consider the groups $l$ with $l\in\set{\iota^*+1,\dots, k}$.
Let $\hat{t} := (\floor[\big]{\frac{m}{2\kappa}}- \floor[\big]{\frac{m}{8\kappa}})$ and $\check{t} := (\floor[\big]{\frac{m}{2\kappa}}- \floor[\big]{\frac{m}{4\kappa}})$.
There are at least $\hat{t}$ jobs in group $\iota^*+1$ and after that alternately exactly $\check{t}$ and $\hat{t}$ many.
Hence, to show $\check{p} \geq m'\sum_{j=0}^{\psi(s^*)-1}p_{k-j}$ it is sufficient to proof that $\check{t}\geq m'$ and $d + \psi(s^*)\hat{t} \leq (k-1)m'$, that is, $\psi(s^*) \leq ((k-1)m' - d)/\hat{t}$.
To show these inequalities, we set $m' = \floor[\big]{\frac{m}{4\kappa}} + \chi_1$ with $\chi_1:= \floor[\big]{\frac{m \bmod 4\kappa}{2\kappa}}$ and note that $\chi_1\in \{ 0,1\}$ as we declared in advance.
This yields $\check{t} = \floor[\big]{\frac{m}{2\kappa}}- \floor[\big]{\frac{m}{4\kappa}} = \floor[\big]{\frac{m}{4\kappa}} + \chi_1 = m'$ (using \cref{rem:exp_and_mod}).
Furthermore, let $\chi_2 = \floor[\big]{\frac{m \bmod 8\kappa}{4\kappa}}$.
Using \cref{rem:exp_and_mod} repeatedly, we get:
\begin{align*}
\frac{(k-1)m' - d}{\hat{t}} & = \frac{ \chi_1(k-1) -  \parenthesis[\bigg]{\floor[\Big]{\frac{m}{2\kappa}} - \floor[\Big]{\frac{m}{4\kappa}} -1} +\parenthesis[\bigg]{\floor[\Big]{\frac{m}{4\kappa}} - \floor[\Big]{\frac{m}{8\kappa}}}\iota(s^*) }{\floor[\Big]{\frac{m}{2\kappa}}- \floor[\Big]{\frac{m}{8\kappa}}}\\
& = \frac{ \chi_1(k-1) -  \parenthesis[\bigg]{2\floor[\Big]{\frac{m}{8\kappa}} + \chi_1 + \chi_2 -1} +\parenthesis[\bigg]{\floor[\Big]{\frac{m}{8\kappa}} +\chi_2 }\iota(s^*) }{3\floor[\Big]{\frac{m}{8\kappa}} + \chi_1 + 2\chi_2}\\
& \geq \frac{ \parenthesis[\bigg]{\floor[\Big]{\frac{m}{8\kappa}} +\chi_2 +\chi_1}(\iota(s^*)-2) }{3\floor[\Big]{\frac{m}{8\kappa}} + \chi_1 + 2\chi_2}\\
& \geq \frac{\iota(s^*)-2}{3} \geq \psi(s^*)
\end{align*}
\end{proof}

The rest of the analysis is carried out separately for the case $s^*>3$ and for the case $s^*=3$.  We start with $s^*>3$ and later provide a similar sequence of arguments for $s^*=3$.

\subparagraph{The rest of the analysis for the case  $\boldsymbol{s^* > 3}$.}
In this case, let $P_1 = \sum_{j=2}^{k-\psi(s^*)}p_j$, $P_2 = \sum_{j=k-\psi(s^*)+1}^{k}p_j$, $\chi_1 = \floor[\big]{\frac{m \bmod 4\kappa}{2\kappa}}$, and $\chi_2 = m \bmod 4\kappa$.
Furthermore, let $\bar{p} = \frac{1}{k - 1 -\psi(s^*)} P_1$ be the average size of the jobs contributing to $P_1$.
Considering the load distribution in the case $s^* > 3$, it is easy to see that the machines from $\floor[\big]{\frac{m}{2\kappa}}+1$ to $m$ each receive $k$ jobs with average size at least $\bar{p}$.
Again using \cref{rem:exp_and_mod} multiple times, we have:
\begin{align*}
\frac{\mathtt{Load}}{m} & \geq \frac{1}{m} \parenthesis[\bigg]{\floor[\Big]{\frac{m}{4\kappa}} +1 } p_1 +
\frac{1}{m}\parenthesis[\bigg]{\floor[\Big]{\frac{m}{2\kappa}} - \floor[\Big]{\frac{m}{4\kappa}}  } P +
 \frac{1}{m} \parenthesis[\bigg]{m - \floor[\Big]{\frac{m}{2\kappa}} } k\bar{p} \\
& \geq \frac{1}{4\kappa} p_1 +
\frac{1}{m}\parenthesis[\bigg]{\floor[\Big]{\frac{m}{4\kappa}} + \chi_1} P +
 \frac{1}{m} \parenthesis[\bigg]{m - 2\floor[\Big]{\frac{m}{4\kappa}} - \chi_1 } k\bar{p} \\
& = \frac{1}{4\kappa} p_1 +
\frac{1}{m}\parenthesis[\bigg]{\frac{m- \chi_2}{4\kappa}  + \chi_1} P +
 \frac{1}{m} \parenthesis[\bigg]{m - \frac{m- \chi_2}{2\kappa} -\chi_1 } k\bar{p} \\
& \geq \frac{1}{4\kappa} p_1 +\frac{1}{4\kappa} P + \frac{2\kappa - 1}{2\kappa} k\bar{p} +
\frac{1}{m}\chi_2P \\
& \geq \frac{1}{4\kappa} p_1 + \frac{1}{4\kappa} P_2 + \frac{1}{4\kappa} P_1 + \frac{2\kappa - 1}{2\kappa} \frac{k}{k - 1 -\psi(s^*)} P_1 \\
& = \frac{1}{4\kappa} p_1 + \frac{1}{4\kappa} P_2 + \parenthesis[\bigg]{\frac{1}{4\kappa}+ \frac{2\kappa - 1}{2\kappa} \frac{k}{k - 1 -\psi(s^*)}} P_1 \\
\end{align*}
We set $x=\frac{1}{4\kappa}+ \frac{2\kappa - 1}{2\kappa} \frac{k}{k - 1 -\psi(s^*)}$ and get:
\begin{align*}
\Alg &= p_1 + P_2 + P_1\\
	 &= p_1 + P_2 + x^{-1}\parenthesis[\bigg]{\frac{1}{4\kappa} p_1 + \frac{1}{4\kappa} P_2 -\frac{1}{4\kappa} p_1 - \frac{1}{4\kappa} P_2  + xP_1}\\
	 &= \parenthesis[\bigg]{1 - \frac{x^{-1}}{4\kappa}}(p_1 + P_2) + x^{-1}\parenthesis[\bigg]{\frac{1}{4\kappa} p_1 + \frac{1}{4\kappa} P_2 + xP_1}\\
	 & \leq \parenthesis[\bigg]{1 - \frac{x^{-1}}{4\kappa}} \Opt + x^{-1} \Opt\\
	 & = \parenthesis[\bigg]{1 + \frac{(4\kappa -1) x^{-1}}{4\kappa}} \Opt	
\end{align*}
Hence, we take a closer look at $x$.
Considering \cref{lem:iota_pattern1}, we can easily see the following:
\begin{claim}\label{claim:migration_psi_bound}
We have $\psi(s,k) \geq \floor[\big]{\frac{k-6}{9}}$ and $k- 1 - \psi(s^*,k) \leq \frac{8k + 6}{9}$.
\end{claim}
\begin{claimproof}
Concerning the first claim, note that:
\begin{align*}
\psi(s,k) &= \floor[\Big]{\frac{\iota(s,k) - 2}{3}}\\
\iota(s,k) &=
\begin{dcases}
\floor[\Big]{\frac{k-1}{3}} & \text{ if }(k-1) \bmod 3 = 1 \text{ and }s \bmod 2 = 1\\
\ceil[\Big]{\frac{k-1}{3}} & \text{ otherwise}\\
\end{dcases}
\end{align*}
It is easy to see that the first value of $k$ for which $\psi(s,k) \geq 1$ holds is $15$ if $s \bmod 2 =0$ and otherwise $k=14$, and $\psi(s,k)$ reliably increments every $9$ steps.
Concerning the second claim:
\[k- 1 - \psi(s^*) \leq k-1 - \floor[\Big]{\frac{k-6}{9}} \leq k-1 - \frac{k-6}{9} + 1 =  \frac{8k + 6 }{9}\]
\end{claimproof}
Since $\psi(s^*) \geq 0$, we have $\frac{k}{k - 1 -\psi(s^*)} \geq \frac{k}{k - 1}$ so we conclude the following.
\[x \geq \frac{1}{4\kappa}+ \frac{2\kappa - 1}{2\kappa} \frac{k}{k - 1} = \frac{2k(2\kappa - 1) + (k-1)}{4\kappa(k-1)} = \frac{4k\kappa - k -1}{4\kappa(k-1)}\]
and furthermore, for $k\leq 14$, we have:
\begin{align*}
\frac{(4\kappa -1) x^{-1}}{4\kappa} &\leq \frac{(4\kappa -1) }{4\kappa}\cdot \frac{4\kappa(k-1)}{4k\kappa - k -1}\\
& = \frac{(4\kappa -1)(k-1)}{4k\kappa - k -1}\\
& = \frac{4k\kappa - 4\kappa -k + 1}{4k\kappa - k -1}\\
& = 1 -  \frac{4\kappa - 2}{4k\kappa - k -1}\\
&\leq 1 -  \frac{4\kappa - 2}{56\kappa - 15}\\
&= 1 -  \frac{4\kappa}{56\kappa - 15} + \frac{2}{56\kappa - 15}\\
&\leq 1 -  \frac{1}{14} + \frac{2}{97}
\end{align*}
Here we used $\kappa \geq 2$ for $s^* > 3$.
Hence, we are done in this case.

If, on the other hand, $k\geq 15$, \cref{claim:migration_psi_bound} yields:
\[\frac{k}{k- 1 - \psi(s,k)} \geq \frac{9k}{8k + 6} \geq \frac{9k}{8k + \frac 25 k} = \frac{15}{14}\]
This yields
\[x\geq\frac{1}{4\kappa}+ \frac{2\kappa - 1}{2\kappa} \frac{15}{14} = \frac{15(2\kappa - 1) + 7}{28\kappa} = \frac{30\kappa - 8}{28\kappa} = \frac{15\kappa - 4}{14\kappa}\]
and furthermore:
\begin{align*}
\frac{(4\kappa -1) x^{-1}}{4\kappa} &= \frac{(4\kappa -1) }{4\kappa}\cdot \frac{14\kappa}{15\kappa - 4}\\
& = \frac{28\kappa -7}{30\kappa - 8}\\
& \leq \frac{28\kappa -7}{29\kappa - 7}\\
& \leq \frac{28}{29} \  .
\end{align*}
And this case works as well.

\subparagraph{The rest of the analysis for the case $\boldsymbol{s^* = 3}$.}
Finally, we make similar considerations for the case $s^* = 3$.
In this case, the machines $\floor[\big]{\frac{m}{2\kappa}}+1$ to $\floor[\big]{\frac{m}{\kappa}}$ receive only one job of size $p_2$ and the argument above does not quite work.
We set $P'_1 = \sum_{j=3}^{k-\psi(s^*)}p_j$, $P'_2 = \sum_{j=k-\psi(s^*)+1}^{k}p_j$, $\chi_1 = \floor[\big]{\frac{m \bmod 4}{2}}$, and $\chi_2 = m \bmod 4$.
Furthermore, let $\bar{p} = \frac{1}{k - 2 -\psi(s^*)} P'_1$ be the average size of the jobs contributing to $P'_1$.

Considering the load distribution for $s^* = 3$, it is easy to see that the machines from $\floor[\big]{\frac{m}{2\kappa}}+1$ to $m$ each receive one job of size $p_2$ and $k-1$ jobs with average size at least $\bar{p}$ and we set $P' = p_2 + (k-1)\bar{p}$.
Note that $P' \geq P$.
We have:
\begin{align*}
\frac{\mathtt{Load}}{m} & \geq \frac{1}{m} \parenthesis[\bigg]{\floor[\Big]{\frac{m}{4}} +1 } p_1 +
\frac{1}{m}\parenthesis[\bigg]{\floor[\Big]{\frac{m}{2}} - \floor[\Big]{\frac{m}{4}}  } P +
 \frac{1}{m} \parenthesis[\bigg]{m - \floor[\Big]{\frac{m}{2}} } P' \\
& \geq \frac{1}{4} p_1 +
\frac{1}{m}\parenthesis[\bigg]{\floor[\Big]{\frac{m}{4}} + \chi_1} P +
 \frac{1}{m} \parenthesis[\bigg]{m - 2\floor[\Big]{\frac{m}{4}} - \chi_1 } P' \\
& = \frac{1}{4} p_1 +
\frac{1}{m}\parenthesis[\bigg]{\frac{m - \chi_2}{4} + \chi_1} P +
 \frac{1}{m} \parenthesis[\bigg]{m - \frac{m}{2} + 2\chi_2 - \chi_1 } P' \\
& \geq \frac{1}{4} p_1 +\frac{1}{4} P + \frac{1}{2}P'\\
& = \frac{1}{4} p_1 + \frac{3}{4}p_2 + \frac{1}{4} P'_2 + \parenthesis[\bigg]{\frac{1}{4}+ \frac{1}{2}\cdot \frac{k-1}{k - 2 -\psi(s^*)}} P'_1 \\
\end{align*}
We set $x=\frac{1}{4}+ \frac{1}{2}\cdot \frac{k-1}{k - 2 -\psi(s^*)}$ and take a closer look at this value.
If $k\leq 14$, we have
\[x \geq \frac{1}{4}+ \frac{1}{2}\cdot \frac{k-1}{k - 2}  = \frac{k-2 + 2(k-1)}{4(k-2)} = \frac{3k-4}{4(k-2)} \geq \frac{19}{24}\]
and, if $k\geq 15$, \cref{claim:migration_psi_bound} yields
\[\frac{k-1}{k-2-\psi(s^*)} \geq \frac{9(k-1)}{8k-3} =\frac{9(k-1)}{8(k-1) + 5} \geq \frac{9(k-1)}{(8 + \frac{5}{14})(k-1)} = \frac{14}{13}\]
and furthermore
\[x \geq \frac{1}{4}+ \frac{28}{52} = \frac{41}{52}\]
and using the fact that $\frac{19}{24} > \frac{41}{52}$, we conclude that in both cases $x\geq \frac{41}{52}$.
Hence, we have:
\begin{align*}
\Alg &= p_1 + p_2 + P'_1 + P'_2\\
	 &= p_1 + P_2 + p_2 + x^{-1}\parenthesis[\bigg]{\frac{1}{4} p_1 + \frac{3}{4}p_2 +\frac{1}{4} P'_2 -\frac{1}{4} p_1 - \frac{3}{4}p_2- \frac{1}{4} P'_2  + xP'_1}\\
	 &\leq p_1 + P_2 + p_2 + \frac{52}{41}\parenthesis[\bigg]{\frac{1}{4} p_1 + \frac{3}{4}p_2 +\frac{1}{4} P'_2 -\frac{1}{4} p_1 - \frac{3}{4}p_2- \frac{1}{4} P'_2  + xP'_1}\\
	 &= \frac{28}{41}(p_1 + P'_2) +\frac{2}{41} p_2 +  \frac{52}{41}\parenthesis[\bigg]{\frac{1}{4} p_1 + \frac{3}{4}p_2 + \frac{1}{4} P'_2 + xP'_1}\\
	 &= \frac{28}{41}\Opt +\frac{1}{41}\Opt +  \frac{52}{41}\Opt = \frac{81}{41}\Opt
\end{align*}
Above, we used one further lower bound for $\Opt$, namely, $\Opt\geq 2p_2$ that is a valid lower bound as there are at least $m+1$ jobs of size at least $p_2$.

\subparagraph{Last Job in Last Phase.}

We consider the case $s^* = \phases$.
In the present case, the borders of the rounds are $1$ and $2$.
Furthermore, the third border of the previous round $\mu(3)\in\set{3,4}$ is relevant as well.
Note that $m\in\set{2,3}$ implies $s^*=3$ and this case has to be dealt with as well.
We get the following load distribution (see also \cref{fig:pattern1_s=xi}):
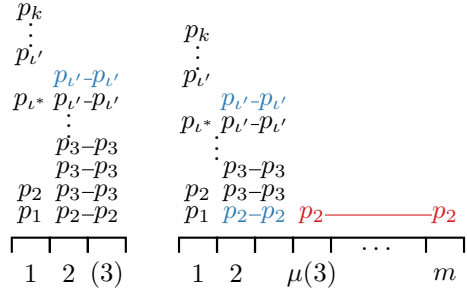
\begin{figure}[h]
\centering
\begin{tikzpicture}
\pgfmathsetmacro{\MachW}{0.5}
\pgfmathsetmacro{\MachH}{0.2}
\pgfmathsetmacro{\MachGap}{2.5*\MachW}
\pgfmathsetmacro{\JobH}{0.3}
\pgfmathsetmacro{\MachLabelShift}{-2.4*\MachH}

\foreach \a/\b/\c/\d in {	0/0/0/$1$, 
							1/0/0/$2$,
							2/0/0/,
							3/0/0/$\mu(3)$,
							4/1/0/$m$
						}
{
	\draw[thick] 	(\a*\MachW + \b*\MachGap,-\MachH) -- 
					(\a*\MachW + \b*\MachGap,0) -- 
					(\a*\MachW + \b*\MachGap + \MachW,0) node[midway,yshift= \MachLabelShift cm, xshift = \c cm]{\d} -- 
					(\a*\MachW + \b*\MachGap + \MachW,-\MachH);
}
\foreach \a/\b in {4/0}{
	\draw[thick] (\a*\MachW + \b*\MachGap,0) -- (\a*\MachW + \b*\MachGap + \MachGap,0) node[midway, yshift = -0.7*\MachH cm ] {$\dots$};
}
\foreach \a/\b/\aa/\bb/\xs/\xxs/\h/\c/\l in 	
{
	0/0/0/0/0.5/0.5/1/black/$p_1$,
	1/0/2/0/0.5/0.5/1/mdarkblue/$p_2$,
	3/0/4/1/0.5/0.5/1/mred/$p_2$,
    0/0/0/0/0.5/0.5/2/black/$p_2$,
    1/0/2/0/0.5/0.5/2/black/$p_3$,
    1/0/2/0/0.5/0.5/3/black/$p_3$,
	0/0/0/0/0.5/0.5/5/black/$p_{\iota^*}$,	
	1/0/2/0/0.5/0.5/5/black/$p_{\iota'}$,	
	1/0/2/0/0.5/0.5/6/mdarkblue/$p_{\iota'}$,	
    0/0/0/0/0.5/0.5/7/black/$p_{\iota'}$,
	0/0/0/0/0.5/0.5/9/black/$p_{k}$
}
{
	\draw[color = \c] (\a*\MachW + \xs*\MachW +\b*\MachGap,\h*\JobH) node[rectangle, fill = white, inner sep=0pt, text = \c] {\l} -- (\aa*\MachW + \bb*\MachGap + \xxs*\MachW,\h*\JobH) node[rectangle, fill = white, inner sep=0pt, text = \c] {\l};
}
\foreach \a/\b/\h in 
{
	1/0/4,
	0.5/0/8
}
{
	\node at (\a*\MachW + \b*\MachGap,\h*\JobH + 0.3*\JobH) {\vdots};
}


\begin{scope}[xshift = -2.2 cm, yshift = 0.0cm]

\foreach \a/\b/\c/\d in {	0/0/0/$1$, 
							1/0/0/$2$,
							2/0/0/$(3)$
						}
{
	\draw[thick] 	(\a*\MachW + \b*\MachGap,-\MachH) -- 
					(\a*\MachW + \b*\MachGap,0) -- 
					(\a*\MachW + \b*\MachGap + \MachW,0) node[midway,yshift= \MachLabelShift cm, xshift = \c cm]{\d} -- 
					(\a*\MachW + \b*\MachGap + \MachW,-\MachH);
}

\foreach \a/\b/\aa/\bb/\xs/\xxs/\h/\c/\l in 	
{
	0/0/0/0/0.5/0.5/1/black/$p_1$,
	1/0/2/0/0.5/0.5/1/black/$p_2$,
    0/0/0/0/0.5/0.5/2/black/$p_2$,
    1/0/2/0/0.5/0.5/2/black/$p_3$,
    1/0/2/0/0.5/0.5/3/black/$p_3$,
    1/0/2/0/0.5/0.5/4/black/$p_3$,
	0/0/0/0/0.5/0.5/6/black/$p_{\iota^*}$,	
	1/0/2/0/0.5/0.5/6/black/$p_{\iota'}$,	
	1/0/2/0/0.5/0.5/7/mdarkblue/$p_{\iota'}$,	
    0/0/0/0/0.5/0.5/8/black/$p_{\iota'}$,
	0/0/0/0/0.5/0.5/10/black/$p_{k}$
}
{
	\draw[color = \c] (\a*\MachW + \xs*\MachW +\b*\MachGap,\h*\JobH) node[rectangle, fill = white, inner sep=0pt, text = \c] {\l} -- (\aa*\MachW + \bb*\MachGap + \xxs*\MachW,\h*\JobH) node[rectangle, fill = white, inner sep=0pt, text = \c] {\l};
}
\foreach \a/\b/\h in 
{
	1.5/0/5,
	0.5/0/9
}
{
	\node at (\a*\MachW + \b*\MachGap,\h*\JobH + 0.3*\JobH) {\vdots};
}

\end{scope}

\end{tikzpicture}
\caption{The load pattern in the case $s^*=\xi$. 
We set $\iota' = \iota^* + 1$.
In the left hand picture, we have $m\in\{2,3\}$ and $s^*=3$, and the sixth row on machines $2$ and $3$ (blue) represents $(k+1) \bmod 3$ many.
Regarding the right hand picture, the first job on the machines $2$ to $\floor[\big]{\frac{2m}{\kappa}}$ (blue) represents several jobs of size $p_2$ and the last job on these machines (also blue) maybe present or not depending on $k$ and $m$.
The jobs on the machines $\mu(3)$ to $m$ (red) represent $k$ jobs of size $p_2$.
} 
\label{fig:pattern1_s=xi}
\end{figure}
\begin{itemize}
\item Machine $1$ receives the sequence $(p_1,\dots,p_k)$.

\item In the case $m\in\set{2,3}$, machines $2$ to $m$ receive the sequence $(p_2, p_3, \dots, p_{\iota^*}, p_{\iota^* + 1})$ with multiplicities $(1,3,\dots,3, 1+ ((k+1) \bmod 3))$.
Note that in the special case $k\in \set{3,4}$, we have $\iota^* = 2$ and the machines receive one job of size $p_2$ and $k-1$ jobs of size $p_3$.

\item If $m>3$, machines $2$ to $\mu(3)-1$ receive the sequence $(p_2,p_3,\dots,p_{\iota^*},p_{\iota^*+1})$ with multiplicities $(1 + \iota(s^*-2,k), 2,\dots,2,x)$ where $x=2$ if $(k-1)\bmod 3 = 0$ or both $(k-1)\bmod 3 = 1$ and $s^*\bmod 2 = 0$, and $x=1$ otherwise (see \cref{cor:even_odd}).

\item Machines $\mu(3)$ to $m$ receive $k$ jobs of size $p_2$.

\end{itemize}
In any case, we get the following lower bound on $\Opt$:
\[\Opt \geq p_1 +\sum_{j=\iota^*+1}^{k}p_j \]
We first consider $m>3$ so $m\geq 4$.
In this case, we have
\[\frac{\mathtt{Load}}{m} \geq \big(2-\frac{1}{m}\big)\sum_{j=2}^{\iota^*}p_j\geq \frac{7}{4}\sum_{j=2}^{\iota^*}p_j\]
which implies
\[\Alg = p_1+ \sum_{j=2}^{\iota^*}p_j +\sum_{j=\iota^*+1}^{k}p_j \leq \frac{11}{7} \Opt\]
completing the analysis of this case.
If, on the other hand, $m\in\set{2,3}$, we have
\[\frac{\mathtt{Load}}{m} \geq p_2 + 2\sum_{j=2}^{\iota^*}p_j\]
which yields
\[\Alg = p_1+ p_2 + \sum_{j=2}^{\iota^*}p_j +\sum_{j=\iota^*+1}^{k}p_j \leq \frac{3}{2} \Opt + \frac{1}{2}p_2 \leq \frac{7}{4} \Opt\]
completing the analysis of this case as well.

We conclude that the rate of our algorithm is at most $\frac{81}{41}$.
\begin{theorem}
There is an ordinal algorithm for cardinality constrained scheduling of rate at most $\frac{81}{41}$.
\end{theorem}

\section{Algorithms with constant migration factors\label{sec:migration}}
In this section, we consider algorithms with constant migration factor.
We start our study by showing a general reduction from ordinal algorithms to algorithms with constant migration factor.
Namely, we prove that if there is an ordinal algorithm of rate at most $\alpha$, then for every $\eps>0$, there is a robust $((1+\eps)\cdot \alpha)$-approximation algorithm whose migration factor is $O(\frac{1}{\eps})$.
Together with the results of \cref{sec:ordinal}, this gives an algorithm of approximation ratio strictly smaller than $2$ and thus improves upon the lower bound on the competitive ratio of pure online algorithms (see \cref{sec:online}).
Then, we show that for every $m\geq 3$, there is no robust PTAS.
This last impossibility result is for values of $k$ that are sufficiently large (and grow unbounded when the input encoding length grows without bound).
Thus, there are two main cases that do not follow this impossibility result.
Namely the case of two machines ($m=2$), and the case of $k$ being a constant term so the migration factor may depend on $k$ as well.
For each of these cases we present a robust PTAS.
In particular, we give a robust FPTAS for the case of two machines and a robust EPTAS for the case of constant $k$.

\subsection{The general upper bound}
Our next goal is to prove the following recipe for transforming an ordinal algorithm into an algorithm with constant migration factor.

\begin{theorem}
Given a polynomial time algorithm $\Alg$ for the ordinal settings of rate at most $\alpha$, there is a robust $((1+\eps)\alpha)$-approximation algorithm whose migration factor is $\frac{1+\eps}{\eps}$.
\end{theorem}
\begin{proof}
Upon the release of a new job $j$, we immediately round up its size $p_j$ to the next integer power of $(1+\eps)$.
Let $p'_j$ be the rounded size of job $j$.
Our algorithm ignores the original sizes of jobs and simply schedules the jobs of this rounded input.
Observe that every feasible solution of the original instance is also a feasible solution of the rounded instance and vice-versa. Furthermore, the cost of a solution in terms of the original instance is at most its cost in terms of the rounded instance, and this last term is again at most $(1+\eps)$ times the cost of the solution with respect to the original instance.
Regarding the migration factor, there may be a multiplicative increase by a factor of $1+\eps$.
Thus, in order to prove the claim, it suffices to  present an $\alpha$-approximation algorithm for the rounded instance whose migration factor is at most $\frac{1}{\eps}$.
Thus, in the remainder of this proof we consider the rounded instance.

Our algorithm maintains a list of the jobs, ordered non-increasingly, that were already released followed by a sequence of jobs of size $0$.
The (already) released jobs are sorted in a non-decreasing order of their sizes.
Based on this ordered list of jobs, we assign the jobs using the ordinal algorithm.
The approximation ratio of the resulting algorithm is at most $\alpha$ (for the rounded input) based on the assumption on the rate of the ordinal algorithm.
Thus, in order to prove the claim it suffices to show that we can maintain this sorted list (and its corresponding schedule) by migrating at most one job of each size that is smaller than the size of the newly arrived job.
To see that, recall that in the rounded instance, all job sizes are integer powers of $1+\eps$, and so the total size of migrated jobs when $j$ is released would be at most $p'_j\cdot \sum_{i=1}^{\infty} \frac{1}{(1+\eps)^i} = p'_j \cdot \frac{1}{\eps}$ as we claimed.

In the rounded instance, we let a \emph{size class} be the set of jobs of a common size, and this appears as a consecutive sublist of jobs in the sorted list.
Observe that we can modify the sorted order of jobs by changing the order of jobs of a common size class (but when reflecting this change to the schedule this may create further migration).
We append the new job $j$ as the smallest job of its size class.
Hence, it will be placed at the position of the largest job of the next smallest non-empty size class.
We can remove and reinsert this job treating it the same way as we did the new job.
Hence, for every non-empty size class whose common size is smaller than $p'_j$ we take its largest job $j'$ and move it to become the smallest of its size class.
As the last step of this procedure, one of the size $0$ dummy jobs is removed from the list.
Observe that when reflected to the schedule, the jobs which were not the largest among their size class were not migrated by this resorting of the jobs.
Furthermore only one job of each such size class is migrated.
The running time of this procedure is linear (for every arriving job) so the claim follows.
\end{proof}
Note that this result can be easily generalized for the dynamic case with job departures (see e.g. \cite{SV16}). 
Since the dynamic case is not within the scope of the present work, this is not discussed further.

\subsection{Lower bound for $\boldsymbol{m\geq 3}$ and non-constant values of $\boldsymbol{k}$}

Let $X=\frac{-3+\sqrt{837}}{2}\approx 12.965476$ be the positive root of the equation
$(18+\frac{X-9}{2})/(X+6)=\frac{X+6}{18}$ so $12<X<14$.  We present a lower bound of $\frac{X+6}{18}$ on the competitive ratio of any robust algorithm for \pr\ on at least three machines (for non-constant values of $k$).  That is, we prove the following result.

\begin{theorem}
Fix a robust algorithm for \pr\ with $m\geq 3$ machines with migration factor $\beta$ that has a constant competitive ratio $\alpha$ for all values of $k$.
 Then,  $\alpha \geq \frac{X+6}{18} \approx 1.05363756$.
\end{theorem}

\begin{proof}
Assume by contradiction that $\alpha< \frac{X+6}{18}$.
We will consider even values of $k$ that satisfy $k\geq \max \{ \beta+2, 7\}$.

The input sequence has two parts.
The first part consists of $m+3$ jobs, where the first three jobs are of size $6$, the next two jobs are of size $9$ and the remaining $m-2$ jobs are of size $X$.
At this point, there is a feasible solution of makespan $18$ that schedules the three jobs of size $6$ to machine $1$, the pair of jobs of size $9$ to machine $2$, and each of the jobs of size $X$ to a separate machine (without other jobs).
In any other schedule, we have that the makespan is at least $X+6$, so if the algorithm does not create the solution we have just described (up to the indexes of the machines), then the input stops and we get a lower bound of $\frac{X+6}{18}$ as a lower bound on the competitive ratio of the algorithm, contradicting our assumption on $\alpha$.
So assume the algorithm has constructed the solution we have identified.

Next, the input continues with
$(m-3)\cdot (k-1)$ jobs, each of which has size $\frac{6}{k-1}$, and $2(k-2)$ jobs, each of which has size
$\frac{X-9}{k-2}$.
Note that we have
\[\beta\cdot\frac{X-9}{k-2}  < \beta\cdot\frac{5\cdot\frac{k-1}{k-2}}{k-1} \leq \beta\cdot\frac{6}{k-1} < 6 \ , \]
since $X<14$, $k\geq 7$, and $k\geq \beta+2$.
Consider the assignment of these jobs to the machines and note that once the second part of the input starts, the jobs of the first part cannot be migrated.
Each machine among machines $3,4,\ldots ,m$ may receive $k-1$ jobs of the second part.

We are left with $ (m-3)\cdot (k-1) + 2\cdot (k-2) - (m-2)\cdot (k-1)=k-3$ jobs of size at least  $\frac{X-9}{k-2}$.
By the pigeonhole principle and since $k-3$ is odd, one of the machines $1$ or $2$ receives at least $\frac{k-4}{2} + 1$ jobs and thus its load is at least $18+ \frac{k-2}{2} \cdot \frac{X-9}{k-2} = 18 + \frac{X-9}{2}$.

On the other hand, there exists a feasible solution of makespan $X+6$ defined as follows.
Each machine out of machines $4,5,\ldots,m$, is assigned one job of size $X$ and $k-1$ jobs, each of which has size $\frac{6}{k-1}$, and thus its load is $X+6$.
Machine $3$ is assigned one job of size $X$ and one job of size $6$.
Each machine out of machines $1,2$ is assigned one job of size $9$, one job of size $6$ and $k-2$ jobs, each of which has size $\frac{X-9}{k-2}$, so their load is again $X+6$.

Thus, the competitive ratio of the algorithm when $k$ grows without bound is at least $(18+\frac{X-9}{2})/(X+6)=\frac{X+6}{18}$, where the equality holds by the condition on $X$, and we get a contradiction to the assumption on $\alpha$ in this case as well for a sufficiently large value of~$k$.
\end{proof}

\subsection{Robust FPTAS for $\boldsymbol{m=2}$ and non-constant values of $\boldsymbol{k}$}
Here we consider the case of two machines and $k>\frac{1}{\eps^2}$, and exhibit the existence of a robust FPTAS for this case of \pr.
Observe that the makespan scheduling on two machines is NP-hard, so clearly our problem on two machines is also NP-hard (by setting the cardinality bound to be the number of jobs in the instance, making this cardinality constraint meaningless).
Therefore, if we are looking for polynomial time algorithms, an FPTAS is the best result we can hope for (unless $P=NP$).

\subparagraph{Preliminaries - the offline problem on two machines.}
In what follows, we assume that we have an offline FPTAS for the problem on two machines (i.e., an FPTAS that may have an unbounded migration factor).
In order to get such an offline FPTAS from the FPTAS's for the knapsack problem with cardinality constraint \cite{Cap00}, consider an input with $t$ jobs with total size $B$, then we add $2k-t$ zero-sized jobs and we apply the FPTAS for the knapsack problem with cardinality constraint asking for a set of jobs of total size at most $B/2$ that consists of exactly $k$ jobs so as to maximize the total size of jobs of this set (i.e., for every job the size of the job is also its value).
Observe that if the optimal makespan to our problem is of value $O$ (where $O\geq  \frac B2$), then the optimal solution for the knapsack problem with cardinality constraint is $B-O$, and a $(1-\eps)$-approximation algorithm for the latter problem gives a solution for our problem of makespan at most $B-(1-\eps)\cdot (B-O) = (1-\eps)\cdot O +\eps B \leq (1+\eps)\cdot O$. So, we get the required offline FPTAS.

\subparagraph{Initial steps of the algorithm once a new job is revealed.}  The size of each arriving job is immediately rounded up to the next integer power of $1+\eps$.  This rounding hurts the approximation ratio by a multiplicative factor of $1+\eps$ and causes no migration, but if we upper bound the migration factor in the rounded instance by $\beta$, then the migration factor with respect to the original instance is upper bounded by $\beta\cdot (1+\eps)$.  So assume without loss of generality that the size of the jobs are already rounded.

Furthermore, once the new job is revealed to the algorithm we invoke the offline FPTAS for our problem and we denote by $\Opt$ the cost of the returned solution by the offline FPTAS.  We keep the value of $\Opt$  non-decreasing so if the FPTAS has found a better solution with the new job, then the $\Opt$ value is kept without modification, but if the cost of the solution has increased then $\Opt$ refers to the new cost.  By the above analysis of the increase of the approximation ratio due to the rounding, the definition of the offline FPTAS, and using the fact that the optimal cost can only increase, we conclude   that $\Opt$ is at most $(1+\eps)^2$ times the optimal cost of the original instance of our problem.

\subparagraph{The structure of schedules maintained by the algorithm.}
The algorithm maintains as an auxiliary data structure three sorted linked lists of the jobs $L_1,L_2,L_3$ each of which is sorted according to their size from largest to smallest.
It also maintains two pointers $p,p'$ where $p$ points to an element of $L_2$ and $p'$ points to an element of $L_3$.
The number of jobs already scheduled in the previous steps of the algorithm is $j-1$ and we are about to schedule job $j$.
In addition, we have a subset of jobs called {\em marked jobs} denoted as $S$, where at some steps of the algorithm we decide to mark a job and then this marking is permanent.
A necessary condition for marking a job $t$ is that the instance already contains at least $k-2/\eps$ jobs each of which of size at least $\frac{p_t}{\eps^2}$.  Maintaining the data structures causes no migration of jobs but the schedule of the jobs is determined by the pointers and the membership of jobs to the lists and so changing the data structures influences the schedule (and as a result jobs are migrated).  In the following description, jobs and positions of the jobs in $L_1,L_2,L_3$ are exchangeable, and when we say job of position $x$, we mean the job that in the current corresponding linked list appears in the $x$-th position.

The jobs in $L_1$ are scheduled according to an offline FPTAS applied on the sub-instance consisting of (only) these jobs.  Let $M_1$ be the machine whose load is not smaller than the load of $M_2$ when considering only the assignment of the jobs in $L_1$.
The jobs in $L_2$ up to and including position $p$ are scheduled on machine $M_2$.
The remaining jobs of $L_2$ are scheduled in a round-robin fashion starting with machine $M_1$, that is, for every integer number $i=1,2,\ldots $, the job in position $p+2i-1$ is scheduled on $M_1$ and the job in position $p+2i$ is scheduled on machine $M_2$ (if there are such positions in $L_2$).
The jobs in $L_3$ up to and including the job in position $p'$ are scheduled on machine $M_2$.
The remaining jobs of $L_3$ are scheduled on machine $M_1$.  We will require that $p'$ points to the last job in $L_3$ unless machine $M_2$ is assigned exactly $k$ jobs.  Furthermore, we require that all jobs in $L_3$ after position $p'$ are not larger than any other job in $L_2$.

\subparagraph{The motivation for this structure and highlights of the analysis.}
 The jobs in $L_1$ are large jobs (these were the set of large jobs when the last one of those arrived). In particular, we will have that $L_1$ has at most $2/\eps$ jobs.  Whenever a new job is announced as a {\em large job} because its size is at least $\eps \cdot \Opt$, we will add it to $L_1$ and we will be able to migrate all jobs of the instance.  The other jobs are called {\em small jobs}.  In the analysis of the approximation ratio, we will have several cases.
In the first case, when $L_3$ is empty and $L_2$ has no job of position larger than $p$, the makespan is attained by the  jobs of $L_1$, so the approximation ratio follows by the one of the offline FPTAS.  In the second case, there are some jobs scheduled by the round-robin step, and in this case the loads of the machines will be nearly balanced and this will prove the approximation ratio.  In the last case, no job is scheduled by the round-robin assignment rule (that is, $p$ points to the last job of $L_2$) and machine $M_2$ has exactly $k$ jobs.  In this last case, machine $M_1$ gets the smallest jobs in the instance on top of the jobs from $L_1$ assigned to that machine.  Since the optimal solution (for scheduling all jobs) needs to assign all jobs, it cannot assign a significantly smaller subset of small jobs to the machine whose load is larger (when considering only the jobs of $L_1$). In this last case, we will use the assumption that $k>\frac{1}{\eps^2}$ to conclude that the load which the algorithm adds to $M_1$ for jobs  of $L_3$ with positions larger than $p'$ is negligible so the approximation ratio will follow.

Thus, in order to be able to use this type of analysis, we need to maintain the structure upon the release of a new job  using a polynomial time algorithm with constant migration factor.
Our algorithm has several cases, and it applies some procedures for guaranteeing its properties.

\subparagraph{The cases of the algorithm for assigning the next job.}  Next, we elaborate on the different cases of the algorithm depending on the size of the new job that has been revealed.

{\bf{Case 1: the new job $\boldsymbol{j}$ has size at least $\boldsymbol{\eps \Opt}$.}}  We will allow migration of all jobs in the instance.  The total size of the jobs is not larger than $\frac{2}{\eps}$ times the size of the new job, so the migration factor will hold.  In order to guarantee the approximation ratio, we redefine the set of large jobs to be the set of jobs of size at least $\eps \Opt$, and we apply the offline FPTAS to schedule the large jobs.  We insert the remaining jobs one by one using the algorithm for the later cases.  If as a result of the insertion of the small jobs, the makespan is not increased, then the guarantee on the competitive ratio holds by the approximation ratio of the offline FPTAS, and otherwise it holds by the guarantee we will establish in the other cases.  Thus, assume without loss of generality that case 1 does not hold, and the new job is a small job of size smaller than $\eps\Opt$.

Next, we explain the procedures and analyze the claims regarding the migration factor when we will apply these procedures during the cases of the algorithm.

\subparagraph{Pushing job $\boldsymbol{t}$ into the round-robin.} If $t$ is in $L_2$ and in position $p$ then, when we say that we push job $t$ into the round-robin, we mean that the pointer $p$ is decreased by one, and then, we apply the following process where  some equal-sized jobs are changing their positions in $L_2$
(this later process is carried out even when the job $t$ is not in one of the $p$ first positions in $L_2$ but it is always a job in $L_2$).
Denote by $p_t$ the size of job $t$ and by $\pi$ the position of $t$ in $L_2$.  Next, we consider the sublist $L'$ of $L_2$ consisting of all jobs of position not smaller  than $\pi$. For every size $q$ of jobs in this sublist, we let $t_q$ be the first job in the sublist of size $q$.  We apply the following modifications to the sublist $L'$ (and to their position in $L$).  For every size $q$ of jobs in the sublist, we move the position of $t_q$ to be the last job of size $q$ in the sublist.  This concludes the description of the procedure.

Observe that the new list resulting from the procedure is also a sorted list of the jobs.  Furthermore, according to the encoding rules, the schedule corresponding to the new list has the following features.  All jobs of $L_1$ or $L_3$ or job in $L_2$ of position smaller than $\pi$ are assigned to the same machine as in the solution prior to the procedure, and for each size $q$ of jobs in the sublist only job $t_q$ perhaps changes its assigned machine (and all other jobs of this size are not migrated).  In order to verify the last property note that for every other job of size $q$ in the sublist, the number of jobs that appear between $\pi$ and this job in the new list is exactly the number of jobs that used to appear between position $\pi$ and this job in the old list $L_2$ (the job $t$ is added to this set of jobs while job $t_q$ is removed from this set of jobs).  Since job sizes are integer powers of $1+\eps$, the total size of migrated jobs when applying the procedure and pushing a job of size $p_t$ into the round-robin is at most $p_t\cdot \sum_{i=0}^{\infty} \frac{1}{(1+\eps)^i}=p_t\cdot \frac{1+\eps}{\eps}$.

\subparagraph{Rebalancing the loads}  Our algorithm tries to maintain the property that the (absolute value of the) difference in the loads of the two machines is at most $\eps\cdot \Opt$ where the load of a machine is defined (for the purpose of this procedure) as the total size of non-marked jobs assigned to the machine.  This property is maintained as long as  there exists at least one job that is scheduled using the round-robin rule.  Furthermore, the algorithm assumes that all jobs of $L_2$ or $L_3$ are of size smaller than $\eps\Opt$.     There are two cases of this procedure, in the first one machine $M_1$ has a larger load, while in the other case, machine $M_2$ has a larger load.  The procedure is invoked only when the difference in the loads of the machines is larger than $\eps\cdot \Opt$ and in the first case it is assumed that some jobs are assigned using the round-robin step while in the second case, we do not make this assumption.

Assume first that the load of $M_2$ is smaller than the load of $M_1$ by more than $\eps\cdot \Opt$.
Here we have several sub-cases.  In the first sub-case, the size of the new job $j$ is at least ${\eps^3}$ times the size of the last job in $L_2$ and the size of the job of position $p'$ in $L_3$.  In this sub-case we move the last two jobs from $L_2$ to $L_3$ and add those jobs in their position according to their size and then, we check if we need to modify the position $p'$ by one place (because the intermediate number of jobs assigned to $M_2$ is now $k-1$ or $k+1$).  If, as a result of this step, the new value of $p$ points to the last job in $L_2$, then we modify the value of $p$ to be the last job in $L_2$ and in this case, no job would be scheduled using the round-robin rule. So, the procedure will stop (and not apply recursively further times).  Using these modifications two (or one) jobs are moved from the round-robin rule to the  assignment rules of $L_3$, so at most one job is moved from machine $M_1$ to machine $M_2$ and at most one job is moved from machine $M_2$ to machine $M_1$.  If there are two migrated jobs and they have the same size, we implement these changes without migrating jobs by replacing their position in  the lists.
Otherwise, since the sizes of these two jobs are rounded to distinct integer powers of $1+\eps$, the total size of the two migrated jobs is at most $\frac{2}{\eps}$ times the amount in which the difference between the two loads (of the two machines) is decreased.  Observe that after the change in the schedule, we still have the case that machine $M_2$ has a smaller load than the load of $M_1$, but if the difference in the two loads is still larger than $\eps\cdot \Opt$, then we repeat the process as much as we need until the first time where either no job is scheduled according to the round-robin rule or the difference between the two loads is at most $\eps\cdot \Opt$.  The total size of migrated jobs by this recursive call to the modification of the schedule is not larger than $\frac{2}{\eps^4}$ times the size of the new job in all recursive calls that ended to be dealt in this sub-case.

In the second sub-case, the size of the new job $j$ is smaller than $\eps^3$ times the size of the last job in $L_2$ but it is of size at least ${\eps^3}$ times the size of job of position $p'$ in $L_3$.  In this case, we move the two jobs of position $p'$ and position $p'+1$ from $L_3$ to $L_2$ and now $p'$ points to the position $p'-1$, but no job is migrated.
We repeat the rebalancing the loads.

In the last sub-case, the size of the new job is smaller than  $\eps^3$ times the size of any small job assigned to $M_2$ and $M_2$ has $k$ jobs.  Thus, we announce $j$ as a marked job and add it to $S$.  Since the loads of the two machines were in difference of at most $\eps\Opt$ before this iteration of adding $j$ to the schedule, the difference between the two loads after announcing $j$ as a marked job is not larger than that.  Thus, the rebalancing the loads ends in the required state if it was called when the load of $M_2$ is smaller than the load of $M_1$.

Next, assume that the load of $M_2$ is larger than the load of $M_1$ by more than $\eps\cdot \Opt$.  We move the two jobs of positions $p'-1$ and $p'$ from $L_3$ to $L_2$ and then increase $p'$ by one.  This action moves the smallest jobs (among the jobs of $L_3$) that were assigned to $M_2$ from $L_3$ to the round-robin and the largest job assigned to $M_1$ among the jobs of $L_3$ is moved to $M_2$.  This procedure is applied when a new job has size of at least the size of the job previously in position $p'$ in $L_3$ (otherwise, it is impossible to get that the load of $M_2$ is higher) and this is the size of the largest job that may migrate as a result of this step. If the earlier position of $p'$ was to the first job in $L_3$, then we move this job to the list $L_2$ and the new value of $p'$ is to the head of the list (before the first job of the list).  As a result of this operation at most one job is migrated from $M_2$ to $M_1$, and at most one job is migrated from $M_1$ to $M_2$.   If these two moved jobs have the same size, we implement these changes without migrating jobs by replacing their position in  the lists.  Otherwise, since the sizes of these two jobs are rounded to distinct integer powers of $1+\eps$, the total size of the two migrated jobs is at most $\frac{2}{\eps}$ times the amount in which the difference between the two loads is decreased.  Observe that after the change in the schedule, we still have the case that machine $M_1$ has a smaller load than the load of $M_2$, but if the difference in the two loads is still larger than $\eps\cdot \Opt$, then we repeat the process as much as we need until the first time where the difference between the two loads is at most $\eps\cdot \Opt$.  The total size of migrated jobs by this recursive call to the modification of the schedule is not larger than $\frac{4}{\eps}$ times the size of the new job that is revealed as the migrated jobs are not larger than the new revealed job.  This completes the description and analysis of the rebalancing the loads operation.

 The proof of the approximation ratio in the other cases is based on the following lemma.
 \begin{lemma}\label{m2lemma}
 A solution obtained by applying the rebalancing the loads procedure satisfies that its makespan is at most $(1+5\eps)\cdot \Opt$.
 \end{lemma}
 \begin{proof}
 First, observe that the total load of the marked jobs is not larger than $2\eps^2$ times the total size of jobs in the instance, thus if we compute the cost of the solution restricted only to the non-marked jobs and show that it is at most $(1+4\eps)\cdot \Opt$ the claim is followed.

Next, assume that at the end of the procedure the difference between the loads of the two machines is at most $\eps\Opt$.  Since the offline solution has makespan at most $\Opt$ and there are only two machines, we get that the makespan of the solution constructed by the algorithm is at most $(1+\frac{\eps}{2})\cdot \Opt$, so the claim holds in this case.

 If the difference between the two loads is higher than $\eps\Opt$, it means that no job is scheduled by the round-robin rule, so $M_2$ gets the largest small jobs until it has exactly $k$ jobs in total.  Furthermore, machine $M_1$ attains the makespan.  Thus, we need to analyze the load of machine $M_1$.  Observe that the number of jobs assigned to machine $M_1$ is $j-k$.  We let $F_1$ be the total size of jobs assigned to $M_1$ among the jobs of $L_1$, and we let $F_2$ be the total size of small jobs assigned to $M_1$.
 Note that the optimal solution (for the rounded instance) needs to allocate the jobs of $L_1$, and among its machines there is a machine $M$ that receives load of at least $\frac{F_1}{1+\eps}$ of jobs among these  jobs.  This machine $M$ has at least $j-k-2/\eps$ additional small jobs, while in our solution $F_2$ is at most the total size of the smallest $j-k$ small jobs.  Note that in the solution of the algorithm machine $M_2$ is assigned $k$ jobs, and at least $k-2/\eps$ of these jobs are small jobs. Thus, each of these small jobs is of size at least the size of a job contributing to $F_2$.  Thus, each small job contributing to $F_2$ has size of at most $\frac{\Opt}{k-2/\eps}$. Thus, having at most $\frac{2}{\eps}$ additional such jobs, increases the load of the machine by no more that $\frac{2\Opt}{\eps (k-2/eps)} \leq 3\eps\Opt$ where the inequality holds by the assumption on $k$.  Thus, the claim holds in this case as well.
 \end{proof}

 \subparagraph{The descriptions of the other cases of the algorithm.}
The other cases apply when the size of the new job is smaller than $\eps\Opt$.

{\bf{Case 2: the new job $\boldsymbol{j}$ has size smaller than $\boldsymbol{\eps} \Opt$ but its size is not larger than the size of the job of position $\boldsymbol{p}$ in $\boldsymbol{L_2}$ or the list $\boldsymbol{L_3}$ is empty and $\boldsymbol{p}$ points to the last job of $\boldsymbol{L_2}$.}}  In this case, we first insert  $j$ into  $L_2$ into its position that is at most $p$ and if $j$ is the smallest job in $L_2$ we increase $p$ by one (in particular $j$ is temporarily assigned to $M_2$). Then check if the load of $M_2$ among the first $p$ jobs of $L_2$ together with the jobs of $L_1$ is not larger than the load of $M_1$ (among the jobs of $L_1$).  If this required property does not hold, then we apply pushing job $p$ into the round-robin over and over again until the first time in which the load (using these jobs of $L_1$ and the prefix of $L_2$ up to position that is the current value of $p$) of $M_2$ will be not larger than the load of $M_1$ (considering only the jobs of $L_1$).  Observe that every job that is pushed into the round-robin has size of at most the size of $j$ as once we push $j$ into the round-robin the required condition is satisfied as it used to hold before the iteration. Thus the total size of jobs that is pushed to the round-robin is at most twice the size of $j$.  Using the analysis of this procedure, the total size of migrated jobs due to these applications is at most $\frac{2(1+\eps)}{\eps}$ times the size of $j$, and this quantity also upper bounds the amount in which we need to decrease the difference between the loads of the two machines (when considering all jobs).  To do that we apply the rebalancing the loads operation.  Thus, the migration factor in this case is $O(\frac{1}{\eps^5})$ and since we have applied the rebalancing the loads operation to obtain this schedule, the approximation ratio of $1+5\eps$ is maintained in this case as well using Lemma \ref{m2lemma}.

{\bf{Case 3: the new job $\boldsymbol{j}$ has size  smaller than the size of the job of position $\boldsymbol{p}$ in $\boldsymbol{L_2}$ but at least the size of the smallest job in $\boldsymbol{L_2}$.}}
In this case, we apply the push job $j$ into the round-robin operation, and this is followed by the rebalancing the loads operation.  Observe that the increase in the difference between the loads of the machines due to the pushing of $j$ into the round-robin is upper bounded by the total size of migrated jobs during this operation (together with the size of $j$) and this is at most $\frac{1+\eps}{\eps}$ times the size of $j$. By the analysis of the rebalancing operation given the fact that the size of each job that we decide to migrate during this operation is at most the size of $j$, we conclude that the total size of migrated jobs due to this step is at most $O(\frac{1}{\eps^5})$  times the size of $j$ and so the migration factor holds in this case as well.  The approximation ratio is maintained by Lemma \ref{m2lemma}.

{\bf{Case 4: Otherwise.}}  We add the new job to $L_3$ (to its position in this list).  If $L_3$ was empty prior to this iteration, then we set $p'$ to point to this job.  Otherwise, if its size is not smaller than the size of job of position $p'$, then $j$ is scheduled (at this point in time) to machine $M_2$ and if its size is smaller than the one in position $p'$, then it is scheduled on machine $M_1$.  If as a result of this step, the number of jobs scheduled to machine $M_2$ is $k+1$, then we move one job to $M_1$.  This is carried out by either  decreasing $p'$ by one with or without moving the smallest job of $L_2$ that is scheduled on $M_1$ to $L_3$, or moving the smallest two jobs of $L_2$ to $L_3$ and assign both these jobs to $M_1$ since their size is not larger than the size of job of position $p'$.  We choose among these possibilities so that the size of every job in $L_2$ will be at most the size of the largest job assigned to $M_1$ among the jobs of $L_3$.  Afterwards, we apply the rebalancing the loads procedure.  Observe that the total size of migrated jobs (including the assignment of $j$) before the rebalancing operation is at most $3p_j$, and thus using the analysis of that operation the migration factor of $O(\frac{1}{\eps^5})$ is kept in this case and by Lemma \ref{m2lemma} the approximation ratio holds as well.

By considering the implication of the initial steps on both the approximation ratio and the migration factor, and then scaling $\eps$ by a constant number, we conclude the following result.
\begin{theorem}
For the case $m=2$ and $k>\frac{1}{\eps^2}$ of \pr, there is a robust FPTAS whose migration factor is $O(\frac{1}{\eps^5})$.
\end{theorem}

\subsection{Robust EPTAS for constant values of $\boldsymbol{k}$}
 In the following, we show an EPTAS  for constant values of $k$ with migration factor $f(1/\eps, k)$, i.e., an algorithm that admits a competitive ratio of $(1+\eps)$ and runs in time $f_1(1/\eps,k)|I|^{O(1)}$ for some computable function $f_1$, an accuracy parameter $\eps$, and the encoding length $|I|$ of the instance $I$. Our scheme has running time that is polynomially bounded even for some non-constant values of~$k$ (e.g. for $k=O(\log \log n)$), but the migration factor is constant only for constant values of~$k$.

In a nutshell, we formulate the problem as a configuration integer-linear program (ILP) and then use sensitivity results for bounding the distance between a solution for the ILP corresponding to the instance before the new job arrives, and the ILP corresponding to the instance after the new job is added to the instance. A {\em configuration} is a multiplicity vector of job sizes. The configuration ILP assigns one configuration onto each machine such that all jobs, i.e., the corresponding job sizes, are covered (see below for precise statement of this configuration ILP).
When a new job arrives, the right-hand side of the ILP (corresponding to the number of present sizes) only changes by one. Due to known sensitivity results, this implies that there exists an optimal solution for the new problem close to the old one. Thus, most of the configurations stay the same. Hence, most jobs are placed as before. Setting up the configuration ILP for the few unplaced jobs yields small ILP dimensions. Thus, the new ILP is solvable efficiently and a migration factor of $f(1/\eps, k)$ is obtained. Further, if $k$ is a constant, this algorithms directly yields a robust EPTAS.

In the following, we first show the necessary preliminary results known such as sensitivity analysis. Then, we formally define the notion of configurations in our scheme together with the configuration ILP for this problem, and finally, present the algorithm together with its analysis.

\subparagraph{Sensitivity analysis of integer programs with respect to a change of the right hand side.}
For a point $x\in \mathbb{R}^{m}$ and a set $Y\subseteq \mathbb{R}^{m}$,
define $\dist(x,Y)$ as the minimal $\ell_{\infty}$-distance of $x$ to any point
in $Y$, i.\,e., $\dist(x,Y)=\min_{y\in Y}\{\lVert x-y \rVert_{\infty}\} = \min_{y\in Y}\{ \max_j \{ | x_j-y_j | \} \}$.

For a given constraint matrix $A\in \mathbb{Z}^{m\times n}$, a right-hand side
$b\in \mathbb{Z}^{m}$, and an objective
function coefficient vector $c\in \mathbb{Z}^{n}$, let $\intSol(A,b,c)$ be the set of optimal solution vectors of the corresponding ILP $ \min\{c^{\top}x\mid x\in
\mathbb{Z}^{n}_{\geq 0}, Ax = b \}$.

The sensitivity analysis with respect to a change of the right hand side of an ILP measures the distance between any optimal integral solution to the closest one for a changed right-hand side. Formally, we define $\sens(A,b,b',c)$ as $$\max_{x \in \intSol(A,b,c)}\{\dist(x, \intSol(A,b',c))\} . $$

\begin{proposition}[Theorem 5 in~\cite{DBLP:journals/mp/CookGST86}]\label{p:CookSensitivity}
  If both $\intSol(A,b,c)$ and $\intSol(A,b',c)$ are non-empty, we have $\dist(x,\intSol(A,b',c))\leq (\lVert b-b'  \rVert_{\infty}+2)\cdot n\cdot \subDet(A)$  for each
  $x\in \intSol(A,b,c)$, where $\subDet(A)$ is the maximum absolute value of a sub-determinant of $A$.
\end{proposition}
This implies $\sens(A,b,b',c)\leq (\lVert b-b'
\rVert_{\infty}+2)\cdot n\cdot \subDet(A)$.
Further, by the Hadamard inequality, $\subDet(A) \leq \Delta^{m}\cdot
m^{m/2}$ where $\Delta$ denotes the largest absolute value of an entry in $A$.  We summarize the use of Proposition \ref{p:CookSensitivity} and Hadamard inequality as follows.
\begin{corollary}\label{cor:sens}
 If both $\intSol(A,b,c)$ and $\intSol(A,b',c)$ are non-empty, we get that $\sens(A,b,b',c)\leq (\lVert b-b'
\rVert_{\infty}+2)\cdot n\cdot \Delta^{m}\cdot
m^{m/2}$ where $\Delta$ denotes the largest absolute value of an entry in $A$.
  \end{corollary}

\subparagraph{The Configuration ILP.}
Denote the current round by $t$, i.e., $t$ is the number of jobs released so far (including the newly released job).  Our scheme has a guessed value of the optimal makespan that we denote by $T$, and with respect to that guessed value we formulate a configuration ILP.
Denote the set of sizes present in round $t$ as $\mathcal P^{(t)}$.
A configuration after the $t$-th job has been released is a non-negative integer vector $\kappa = (k_1^\kappa, k_2^\kappa, \dots, k_{|\mathcal P^{(t)}|}^\kappa)$ indexed by the set of sizes of the jobs with index at most $t$.  We say that a machine has configuration $\kappa$ if for all $j$, $k_j^\kappa$ jobs with size $p_j \in \mathcal P^{(t)}$ are scheduled on that machine.

We call a configuration $\kappa \in \mathcal K^{(t)}$ \emph{valid} if $\sum_{j=1}^{|\mathcal P^{(t)}|} k_j^\kappa \leq k$ satisfying the cardinality constraints and $\sum_{j=1}^{|\mathcal P^{(t)}|} p_j k_j^\kappa \leq T$ fulfilling the makespan for some makespan guess~$T$. Otherwise, the configuration is \emph{non-valid}.
Denote the set of valid configurations of round $t$ as $\mathcal K^{(t)}$.

Define the variable $x_\kappa$ for the occurrence of valid configuration~$\kappa$. Denote by $a_j$ the number of jobs with size $p_j$ in the instance up to and including the $t$-th job. The configuration ILP~(config-ILP$)^{(t)}$ of round $t$ is the feasibility integer linear program stated as follows:

\begin{align*}
\tag{1}
\label{1}
& \sum_{\kappa \in \mathcal K^{(t)}} x_\kappa = m \\
\tag{2}
\label{2}
& \sum_{\kappa \in \mathcal K^{(t)}} x_\kappa \kappa_j^\kappa = a_j & \forall j = 1, \dots, |\mathcal P^{(t)}| \\
& x_{\kappa}\geq 0& \forall \kappa \in {\mathcal K^{(t)}}
\end{align*}

Constraint \eqref{1} assures that we use exactly one configuration for each machine. The constraint \eqref{2} is satisfied if all present jobs (job sizes) are covered, so that every job is assigned to exactly one machine. It holds that (config-ILP$)^{(t)}$ has $|\mathcal P^{(t)}| + 1$ rows and $|\mathcal K^{(t)}|$ columns.

We upper bound the number of valid configurations by $(k+1)^{|\mathcal P^{(t)}|}$, that is a valid upper bound as each job size can occur zero up to at most $k$ times. All entries in the constraint matrix are non-negative and their largest number is  $k_j^\kappa \leq k$ for some sizes $p_j$ and configuration $\kappa$. The largest component of the right-hand side vector is at most $\max\{m, \max_j\{a_j\}\} \leq t$, as we have $t$ jobs in round $t$ and cannot have more jobs of a certain size than the total number of jobs.  These upper bounds are used only for bounding the migration factor and time complexity of our scheme, the algorithm itself can compute the exact number of valid configurations, the maximum component of the constraint matrix, and for two right hand sides, it can compute the upper bound established in Proposition \ref{p:CookSensitivity}.

\subparagraph{The Algorithm.}  We next present the pseudo-code of the algorithm to schedule the $t$-th job and migrating some other jobs, then define formally the different steps of the algorithm, and later analyze its approximation ratio and its migration factor.  We let $\sigma^{(t)}$ denote the schedule after scheduling job $t$.

Place the first $m$ jobs onto different machines, the $j$-th job is assigned to machine $j$ (for $1\leq j \leq m$). Afterwards, each time a new job $t$ arrives, execute the following steps:
\begin{itemize}
\item[(1)] compute a lower and upper bound on the optimal makespan;
\item[(2)] if the size of the new job $t$ is smaller than $\delta/k$ times the current lower bound:
\begin{itemize}
\item[(2.1)] job $t$ is a small job and we place it onto an arbitrary machine among the machines that before this assignment have less than $k$ jobs.
\end{itemize}
\item[(3)] otherwise ($t$ is a large job):
\begin{itemize}
\item[(3.1)]round the sizes of the jobs;
\item[(3.2)] interpret the previous schedule (of the prefix of jobs not including job $t$) as a solution of (config-ILP$)^{(t)}$ and identify the jobs that remain untouched using the sensitivity result;
\item[(3.3)] set up and solve the reduced (config-ILP$)^{(t)}$ for the remaining jobs and assign them using the solution of the reduced (config-ILP$)^{(t)}$ onto the machines.
\end{itemize}
\end{itemize}
The step (3.3) requires that we know the desired (optimal) makespan. To circumvent this, we embed the last step in a binary search between the current lower and upper bound and take the lowest, feasible guess.

In the following, we go trough each step in more detail and directly argue their properties. Let $\eps$ be the desired approximation ratio. Further, assume some accuracy value~$\delta$ we specify in dependence of $\eps$ later on. Obviously, placing the first $m$ jobs onto different machines is optimal. Thus, let us assume that we have placed at least $m$ jobs, i.e., the next new job is the $t$-th job with $t \geq m+1$.  The following steps are then executed.

 \emph{Step (1) - Computing a lower and upper bound.} An estimation for the lower bound $\text{LB}^{(t)}$ of round~$t$ corresponds to an equal distribution of the overall size of all $t$ present jobs, but in some cases we decide to increase the lower bound beyond this value.  Since it is important that this value is not decreased, we use the following lower bound, $\text{LB}^{(t)} = \max\{ \text{LB}^{(t-1)},  \lceil \sum_{j=1}^t p_j/m \rceil , p_{\max}^{(t)}\}$ where $\text{LB}^{(0)}=0$, and $p_{\max}^{(t)}$ denotes the maximum size of a job up to and including job $t$. As the makespan is always integral, we can safely round the value up to obtain an integer one. An upper bound for round~$t$ places the largest $k$ jobs onto one machine, and we upper bound this quantity by the following upper bound $\text{UB}^{(t)} = k \cdot p_{\max}^{(t)}$. Note that the arrival of job $t$ may change these bounds, but when new jobs are released, these bounds cannot decrease, so they either stay the same or they increase. When the new job is released, we recalculate the two bounds in step (1).

 \emph{Step (2) - Assignment of small jobs.}  If the size of the new job $t$ is smaller than $\frac{\delta}{k} \cdot \text{\,LB}^{(t)}$, we say that job $t$ is a {\em small job} and we  place it onto any machine which has less than $k$ jobs. As job $t$ is small, this only produces a small error.  That is, in the solution constructed by the algorithm, the machine attaining the makespan will have a total size of small jobs assigned to it of at most $\delta$ times the current lower bound using the fact that the lower bounds cannot decrease.

Otherwise, we proceed with the following steps.

 \emph{Step (3.1) - Rounding.}  For each present job size $p_j < \delta/k \text{\,LB}^{(t)}$, set the job size to $p_j' = 0$.  Otherwise, for all $p_j \geq \delta/c \cdot \text{LB}^{(t)}$, we round the sizes geometrically, that is, job~$j$ has a rounded size of $p_j' = \lceil (1+\delta)^\ell/\delta \rceil \cdot \delta^2/c \cdot \text{LB}^{(t)}$ for the integer value of~$\ell$ satisfying $\delta/c \cdot  \text{LB}^{(t)} (1+\delta)^{\ell -1} < p_j \leq \delta/c \cdot  \text{\,LB}^{(t)} (1+\delta)^{\ell}$. Note that the rounded job size of a large job is at most $1+\delta$ times its size and at least the size of that job in the instance before the rounding.  Furthermore, in future iterations of the algorithm, the rounded size of the job stays the same until the first iteration where it is dropped to zero. Denote the set of (different) rounded job sizes of round $t$ by $\mathcal P^{(t)}$ where $|\mathcal P^{(t)}| \leq O(1/\delta \log(k/\delta))$. Further, we can use the compact representation of the jobs in the rounded instance as a vector  $(a_1, \dots, a_{|\mathcal P^{(t)}|})$ where $a_j$ states the number of jobs with sizes $p_j' \in \mathcal P^{(t)}$.

\emph{Step (3.2) - Applying sensitivity.} First, we interpret the schedule $\sigma^{(t-1)}$ of the previous iteration without inserting job~$t$ as a feasible solution of (config-ILP$)^{(t)}$ for the lower bound of iteration $t$  and for the instance after the rounding of iteration $t$. To do so, first replace the job sizes of each job in the schedule $\sigma^{(t-1)}$ by its current rounded size (that is an element of  $\mathcal P^{(t)}$). Recall that the new rounded size of a job is not larger than the rounded size of the same job in the previous iteration.  Then, we define for each machine the configuration of that machine (with respect to the rounded instance).  Last, for each configuration $\kappa \in \mathcal K^{(t)}$, we let  $x_\kappa$ be the number of machines of this configuration. As we have no objective function, feasibility corresponds to optimality (i.e., we define a dummy objective function vector $c$ that is the zero-vector, and then, every feasible solution for the ILP is also an optimal solution of that ILP).

Note that these configurations, which result from $\sigma^{(t-1)}$, i.e., from the schedule of the previous round, stay valid during the current round $t$: This is so as the lower bound can only increase and the makespan of the rounded instance (ignoring job $t$) can only decrease.

We solve the configuration ILP after introducing job $t$ with its rounded size in the desired time complexity as well.  If the configuration ILP is infeasible, then we increase the value of $\text{\,LB}^{(t)}$ by a multiplicative factor of $1+\delta$ and go back to step (2) without increasing the value of $t$.  Since the ratio between the valid upper bound on the optimal makespan and the value of the lower bound is at most $k$, we are guaranteed that within $O(1/\delta \log(k/(1+\delta)))$ iterations, we will have an iteration where the ILP is feasible and then, we continue to the remaining of this step (3.2).

Next, we compare the two configuration ILPs, both of which are with respect to the rounding carried out in step (3.1) of round $t$, the first is without job $t$, and the second takes job $t$ into account as well.  First observe that the constraint matrix is the same matrix for the two ILPs.
Due to the arrival of job $t$, the right-hand side only changes slightly: The corresponding component $a_j$ for which the $(j+1)$-th constraint in the configuration ILP counts the number of jobs of size $p'_t$ that are scheduled in the configurations of the machines,  is increased by one and all other components of the right hand side vector stay the same. Hence, by Corollary \ref{cor:sens}, we get that there exists an optimal solution for the altered problem with distance at most
\begin{align*}
&3|\mathcal K^{(t)}| \cdot k^{|\mathcal P^{(t)}|+1} \cdot (|\mathcal P^{(t)}|+1)^{(|\mathcal P^{(t)}|+1)/2}\\
& \leq  3 (k+1)^{|\mathcal P^{(t)}|} \cdot  k^{|\mathcal P^{(t)}|+1} \cdot (|\mathcal P^{(t)}|+1)^{(|\mathcal P^{(t)}|+1)/2}\\
& \leq 3 (k+1)^{(1/\delta \log(k/\delta))} \cdot (1/\delta \log(k/\delta))^{(1/\delta \log(k/\delta))+1}
 \in O((k/\delta)^{O(k/\delta \log^2(k/\delta)))} \ . \\
\end{align*}
The algorithm computes the value $Dist=3|\mathcal K^{(t)}| \cdot k^{|\mathcal P^{(t)}|+1} \cdot (|\mathcal P^{(t)}|+1)^{(|\mathcal P^{(t)}|+1)/2}$ exactly, and find a feasible solution $x$ to the configuration ILP (config-ILP$)^{(t)}$ of distance at most $Dist$ from the solution for the modified right hand side corresponding to $\sigma^{(t-1)}$.
 Next, we  leave $x_\kappa' = \max\{\lceil x_\kappa - Dist \rceil, 0\}$ many configuration $\kappa$ untouched for each $\kappa \in \mathcal K^{(t)}$.  These configurations are assigned to the machines, but in this assignment, we require that for every machine, we either keep the configuration of that machine (as it used to be based on the schedule $\sigma^{(t-1)}$) or we leave the machine temporarily without a configuration.  The jobs that were assigned in the schedule  $\sigma^{(t-1)}$ to a machine for which we have assigned a configuration are left assigned to the same machine, and these machines will not receive a new job.  The jobs previously assigned to a machine for which we temporarily leave without a configuration (at most $k$ jobs per such machine) are the jobs that may migrate.  We denote by $J'$ this subset of jobs that may migrate and we add $t$ to this set of jobs.

\emph{Step (3.3) - Scheduling the remaining jobs of $J'$ to $M'$.} We have to pack the remaining jobs of $J'$ onto $m' = |M'|$ left over machines, and we denote the set of leftover machines by $M'$.
To do that we use the solution $x$ to the configuration ILP which we found in step (3.2), and we allocate a configuration to each machine in $M'$.  We do so in a way that together with the machines for which we had configurations already, for every $\kappa$ the total number of machines assigned configuration $\kappa$ is exactly $x_{\kappa}$.  Since $x$ is a feasible solution for the ILP, $\sum_{\kappa} x_{\kappa}=m$ and so we can do that and every machine has exactly one configuration.  We allocate the  jobs of $J'$ to the machines in $M'$ based on the configurations of these machines, so that if the assigned configuration of a machine in $M'$ is $\kappa$ then for every job size in the rounded instance that is encoded in the configuration $\kappa$ we have exactly the number of jobs of this size assigned to that machine.

For the proof of the next theorem, we let $\eps=(1+\delta)^2-1 = 2\delta+\delta^2 \leq 3\delta$.

\begin{theorem}
The algorithm computing at each iteration $t$ the schedule $\sigma^{(t)}$ for the \pr\  is a robust EPTAS for constant values of $k$ with a constant migration factor.
\end{theorem}
\begin{proof}
We first argue that if the current value of the lower bound when we apply step (2) (and the later steps) is at least $(1+\delta)$ times the optimal objective function value of the original (not the rounded one) instance up to and including job $t$, then the configuration ILP has a feasible solution.  This is so, as in the rounded instance, the optimal solution satisfies that every machine has jobs of total size of at most $(1+\delta)$ times their sizes in the original instance (by the rounding rule) and so interpreting the set of jobs assigned to a common machine as a configuration, results in a valid configuration for this value of the lower bound.

The feasibility of the solution returned by the algorithm is guaranteed as if the configuration ILP is feasible, then the algorithms finds a feasible assignment of jobs to configuration, and otherwise, there is no feasible solution for the configuration ILP for the current value of the lower bound and this is impossible if the current lower bound is at least the value of $(1+\delta)$ times the cost of the optimal solution for the original instance.

Next, we prove the approximation ratio of our algorithm.  If $t\leq m$, then the algorithm has placed at most one job per machine, and this is an optimal solution.  Thus, assume that $t>m$, and by induction on $t$, assume that the solution $\sigma^{(t-1)}$ has makespan of at most $(1+\delta)^2$ times the optimal makespan for the instance $I^{(t)}$.  Based on the above argument and the fact that the optimal makespan of a feasible solution cannot decrease if a new job is released, it suffices to show that if the algorithm returns a solution $\sigma^{(t)}$ of a higher makespan than the one of $\sigma^{(t-1)}$, then its new makespan is at most $(1+\delta)$ times the lower bound  $\text{\,LB}^{(t)}$ for the first iteration for which the configuration ILP is feasible. Observe that the solution of the configuration ILP results in a solution to the rounded instance (of the same rounding) of cost that is upper bounded by the maximum total size of jobs assigned by a valid configuration, that is by  $\text{\,LB}^{(t)}$.  When we consider the same solution as a solution to the original instance, the size of every large job is not increased, while the size of every small job is at most $\frac{\delta}{k} \cdot \text{\,LB}^{(t)}$, but since every machine is assigned at most $k$ jobs (so it is assigned at most $k$ small jobs), the resulting makespan of $\sigma^{(t)}$ is at most $(1+\delta) \cdot \text{\,LB}^{(t)}$.  Thus, the approximation ratio is at most $(1+\delta)^2 = 1+\eps$ as required.

Next, consider the running time of our scheme, that is, the time complexity of  the algorithm in iteration $t$ for placing job $t$ and modifying the schedule of some of the other jobs. It takes time $O(t)$ to compute the bounds and round the instance, as we have to go trough each size of the $t$ present jobs once.  The number of iterations of increasing the current lower bound until we get a value for which the configuration ILP results in a feasible solution is $O(\log_{1+\delta} k)$ by the ratio between the upper bound and the lower bound. The main argument in the running time for a fixed iteration after the rounding is the solution of the configuration ILP (first of the rounded instance in step (3.2) and then the reduced one in step (3.3)) while the other steps take $O(kt)$ time.  We are going to use the algorithm of \cite{DBLP:conf/innovations/JansenR19} to do both tasks.  For both ILPs,   the right-hand side is bounded by $n$, the constraint matrix is an integer matrix where every entry is a non-negative integer of at most $t$, and there are $|\mathcal P^{(t)}|+1$ constraints.  The resulting time complexity is less than $(|\mathcal P^{(t)}| k)^{O(|\mathcal P^{(t)}|)} \log(t) +t$, and using the fact that $|\mathcal P^{(t)}| \leq O(1/\delta \log(k/\delta))$, we conclude that our scheme is indeed an EPTAS.

 It remains to analyze the migration factor of our scheme.
If $t\leq m$ or job $t$ is a small job with respect to the value of $ \text{\,LB}^{(t)}$ for which the algorithm is able to find a feasible schedule, then no job is migrated and the claim regarding the migration factor clearly holds.
Assume that the last simple cases do not hold for job $t$.  Then, the algorithm migrate the set of jobs $J'$ that were assigned to machine in $M'$.  The number of machines in $M'$ is at most $Dist \cdot | \mathcal K^{(t)}| \leq 3(|\mathcal K^{(t)}|)^2 \cdot k^{|\mathcal P^{(t)}|+1} \cdot (|\mathcal P^{(t)}|+1)^{(|\mathcal P^{(t)}|+1)/2}$, and each such machine in $M'$ used to have jobs of total size at most $(1+\delta) \cdot \text{\,LB}^{(t)}$.  Thus, the migration factor is at most $(1+\delta) \cdot (k/\delta) \cdot 3(|\mathcal K^{(t)}|)^2 \cdot k^{|\mathcal P^{(t)}|+1} \cdot (|\mathcal P^{(t)}|+1)^{(|\mathcal P^{(t)}|+1)/2}$ and this is upper bounded by a constant for every fixed value of $\eps$ (and thus, the corresponding fixed value of $\delta$, for every constant value of $k$).
\end{proof}

\section{Concluding remarks on the case of Class Constraint Scheduling\label{sec:ClCS}}
A closely related question asks to schedule $n$ jobs, where each job~$j$ admits some class $c_j \in \{1, \dots, C\}$, on $m$ machines, each with a class restriction of $k$. Thus, in contrast to before, the restriction for the machines does not limit the number of jobs scheduled on them but limits the number of classes the jobs belong to. Again, we aim to find a schedule that minimizes the makespan. This problem, called Class Constraint Scheduling (ClCS), is a generalization of \pr\ and it seems rather hopeless to handle in the pure online case.  In fact, for this problem, even constant migration factors do not assist the algorithm.

\subparagraph{Lower Bound on the competitive ratio of robust algorithms for ClCS on identical machines.}
A lower bound of $m$ holds as follows. The input starts with $m$ jobs each of size $1$ that belong to a common class. They have to be placed onto one machine so that no additional class slots are occupied on different machines which we might need later on and cannot free as following jobs might have a size smaller than $1$ over the migration factor (in any other assignment, we may end up reporting an infeasible instance although there is a feasible solution). However, if the input ends after these $m$ jobs, it would have been optimal to place one job onto each machine to obtain the optimal makespan of $1$. This yields a lower bound of $m$. The lower bound of $m$ is tight as a greedy algorithm that assigns the classes to machines so that each machine has at most $k$ classes and then, every job of a given class is assigned to the (unique) machine where this class is assigned is $m$-competitive.  This holds, as the algorithm finds a feasible solution whenever the number of classes is at most $mk$ and its makespan is not larger than the total size of all jobs that is at most $m$ times the optimal makespan.

\subparagraph{Lower Bound on the competitive ratio of robust algorithms for ClCS on uniform machines.}
Let $\beta$ be the migration factor of a given algorithm.
Let the speed values be $s_1 = 1$ for the first machine and  $s>1$ for the remaining $m-1$ ones, i.e.,  $s = s_2 =  \dots = s_m$.
The input starts with $mk$ unit-sizes jobs, each with a unique class $c_1, \dots, c_{mk}$. They have to be distributed among the machines such that each machine gets $k$ different jobs. W.l.o.g. let us suppose that the first machine is assigned the jobs of classes $c_1, \dots, c_k$ and let $k'=\min\{ k,m-1\}$.  Let $\eps>0$ be a small number and let $M>0$ be a large integer number.
Next, we receive $M\beta$ rounds. In each round, $k'$ jobs arrive each of size $1/\beta-\eps$ where for each class among $c_1, \dots, c_{k'}$, there is exactly one job of each round of the class. In an optimal solution, the classes $c_1, \dots, c_{k'}$ would be distributed among the $k'$ fastest machines. For a sufficiently large value of $M$, this would yield a makespan of at most $M/s + k/s$.
In the schedule of the algorithm, the jobs of size smaller than $1/\beta$ are all placed onto the slowest machine so the makespan is $k + k'M\cdot (1-\beta\eps)$. Due to the small sizes of the second phase jobs, no migration is allowed of the first phase jobs. Thus, the competitive ratio of an algorithm with competitive ratio $\beta$ is obtained for $\eps$ tending to zero and $M$ grow unbounded.  The resulting lower bound is  $sk'=s\cdot \min\{k,m-1\}$.

\bibliography{lib}

\end{document}